\documentclass{article}

\usepackage{hyperref}

\usepackage{amsmath,amsthm, amssymb, amsfonts,accents,nccmath}
\usepackage{enumitem,lipsum,array,mathtools}
\usepackage[toc,page]{appendix}
\usepackage[english]{babel}
\usepackage{dsfont}
\usepackage{booktabs}
\usepackage{tikz}
\usetikzlibrary{calc, chains, positioning}
\usetikzlibrary{decorations.pathreplacing,calligraphy}
\usepackage{tikz-network}
\usetikzlibrary{arrows.meta, positioning, shadows}
\usetikzlibrary{hobby}
\usetikzlibrary{hobby}
\usepackage{soul}
\usepackage{pst-node}
\usepackage{tikz-cd}

\usepackage[accepted]{icml2025}

\usepackage{caption}
\usepackage{subcaption}
\usepackage{blkarray}

\providecommand{\E}{\mathbf{E}}
\providecommand{\rot}{\intercal}

\newcommand\restr[2]{{% we make the whole thing an ordinary symbol
		\left.\kern-\nulldelimiterspace % automatically resize the bar with \right
		#1 % the function
		\vphantom{\big|} % pretend it's a little taller at normal size
		\right|_{#2} % this is the delimiter
}}

\makeatletter
\def\BState{\State\hskip-\ALG@thistlm}
\makeatother
\numberwithin{equation}{section}

\newtheorem{theorem}{Theorem} 

\newtheorem{lem}[theorem]{Lemma}

\newtheorem{defn}[theorem]{Definition}

\newcommand{\vol}{\textup{vol}}
\newcommand{\out}{\textup{out}}
\newcommand{\ino}{\textup{in}}

\newcommand{\ds}{\displaystyle}

\author{}

\usepackage[textsize=tiny]{todonotes}

\icmltitlerunning{Online Sparsification of  Bipartite-Like Clusters in Graphs}

\begin{document}

\twocolumn[
\icmltitle{Online Sparsification of  Bipartite-Like Clusters in Graphs}

\begin{icmlauthorlist}
\icmlauthor{Joyentanuj Das}{ed}
\icmlauthor{Suranjan De}{ed}
\icmlauthor{He Sun}{ed}
\end{icmlauthorlist}

\icmlaffiliation{ed}{School of Informatics, University of Edinburgh, Edinburgh,  United Kingdom}

\icmlcorrespondingauthor{He Sun}{h.sun@ed.ac.uk}

\icmlkeywords{Machine Learning, ICML}

\vskip 0.3in
]



\printAffiliationsAndNotice{}  % leave blank if no need to mention equal contribution

\begin{abstract}

  Graph clustering is an important algorithmic technique for analysing massive graphs, and has been widely applied in many research fields of data science. While the objective of most   graph clustering algorithms is to find a vertex set of low conductance,   a sequence of recent studies   highlights the importance of the inter-connection between vertex  sets when analysing real-world datasets. Following this line of research, in this work we study bipartite-like clusters and present efficient and  online sparsification algorithms that  find such clusters in both undirected graphs and directed ones. We   conduct experimental studies on both synthetic and real-world datasets, and show that our algorithms significantly speedup the running time of existing clustering algorithms while preserving their effectiveness.
\end{abstract}

\section{Introduction}
 Graph clustering is a fundamental technique in data analysis with wide-ranging applications in machine learning and data science. A classical graph clustering problem involves partitioning the vertices of a graph into sets of highly connected vertices to minimize the normalised cut value. However, many real-world clustering tasks are defined by alternative objective functions, tailored to the specific needs and constraints of the problem at hand. One such example involves uncovering the vertex sets  that are densely connected to each other, and these two vertex sets form a bipartite-like graph.  For example, when representing the migration or trade datasets with a graph, a bipartite-like cluster   captures a typical pattern of  regional migration or trade~\citep{CLSZ,LaenenS20,HRC22}, and the significance of bipartite-like clusters is  further  highlighted when studying  many other   real-world datasets~\citep{Bennett,chain_network}.

\paragraph{Our Results.}   
We first  study bipartite-like clusters   in   undirected graphs, and present an algorithm that sparsifies an undirected graph while preserving its structure of bipartite-like clusters. Our algorithm can be implemented  online, and directly applied to speed 
  up the running time of   existing algorithms that find   bipartite-like clusters. Formally speaking, for an   undirected  $G=(V,E)$ and 
 a pair of disjoint and non-empty  subsets $A,B \subset V$, we define $$\overline{\phi}_G(A,B) \triangleq \frac{2 w_G(A,B)}{\vol_G(A \cup B)},$$
 where $$
 w_G(A, B) \triangleq \sum_{\substack{ \{u,  v\}\in E \\ u\in A, v\in B }} w(u,v)$$ is the cut value between $A$ and $B$ and $\vol_G(A\cup B)$ is the volume of $A\cup B$  defined by $$\vol_G(A\cup B) = \sum_{\substack{ \{u, v\}\in E \\ u\in A\cup B}} w_G(u,v).$$ Notice that  a high value of $\overline{\phi}_G(A,B)$ implies that most edges adjacent to the vertices in $A\cup B$ are between $A$ and $B$, and $A$ and $B$ form a bipartite-like cluster. 
Generalising this to  multiple clusters, for every $k \in \mathbb{N}$ we define the $k$-way dual Cheeger constant  by
	\begin{equation}
		\bar{\rho}_G(k) \triangleq \max_{(A_1,B_1),\ldots,(A_k,B_k)} \min_{1 \le i \le k} \overline{\phi}_G (A_i,B_i),
	\end{equation}
	where the maximum is taken over
	all the possible  $k$ pairs of subsets $(A_1,B_1),\ldots,(A_k,B_k)$ satisfying $A_i \cap A_j=\emptyset, B_i \cap B_j=\emptyset, A_i \cap B_j = \emptyset$
	for different $i,j\in[k]$, and 
	$A_i\cup B_i \ne \emptyset$
	for different $i,j\in[k]$. Notice that a high value of $\bar{\rho}_G(k)$ implies that $G$ contains $k$ bipartite-like clusters, in each of which the vertex sets $A_i$ and $B_i$ are densely connected to each other.      
 We prove that, when $G$ 
 presents a clear structure of   $k$ bipartite-like clusters,  this structure can be represented by a sparse subgraph $G^*$ of $G$ with $\widetilde{O}(n)$ edges, and $G^*$ can be constructed online in nearly-linear time\footnote{We say that a graph algorithm runs in nearly-linear time if the algorithm's running time is $O(m \cdot\mathrm{poly}\log n)$, where $m$ and $n$ are the number of edges and vertices of the input graph. For simplicity, we use $\widetilde{O}(\cdot)$ to hide a poly-logarithmic factor of $n$.}. Our result is as follows: 
 
 \begin{theorem}[Result for undirected graphs]\label{thm:undirected_main} 
 Let $G = (V_{G},E_{G},w_{G})$ be an undirected and weighted graph of $m$ edges, and assume that $G$ contains $k$   bipartite-like clusters  $(A_1,B_1), \ldots, (A_k, B_k)$ corresponding to  $\bar{\rho}_{G}(k)$. Then, there is an algorithm that runs in  $\widetilde{O}(m)$ time and 
 computes a sparsifier  $G^* = (V_{G},F \subset E_{G}, \widetilde{w})$,  such that    these $k$ bipartite-like clusters are preserved in $G^*$ with high probability. That is, it holds with high probability that $\bar{\rho}_{G^*}(k) = \Omega\left(\bar{\rho}_{G}(k)\right)$, and $G^*$  contains only $k$ bipartite-like clusters. 
	\end{theorem}

Secondly, we study the bipartite-like clusters    in directed graphs. Let   $\overrightarrow{G} = (V_{\overrightarrow{G}},E_{\overrightarrow{G}},w_{\overrightarrow{G}})$ be a   digraph with weight function $w_{\overrightarrow{G}}:E_{\overrightarrow{G}} \rightarrow\mathbb{R}_{\geq 0}$. For any vertex $u \in V_{\overrightarrow{G}}$, we use $\deg_{\out}(u)\triangleq \sum_{(u, v)\in E } w_G(u,v)$ and $\deg_{\ino}(u)\triangleq \sum_{(v, u)\in E } w_G(v,u)$ to express the sum of weights of directed edges with $u$ as the tail or the head, respectively. For any $S \subset V_{\overrightarrow{G}}$, we define $\vol_{\out}(S) \triangleq \sum_{u \in S} \deg_{\out}(u)$ and $\vol_{\ino}(S) \triangleq \sum_{u \in S} \deg_{\ino}(u)$.  For any two disjoint subsets $A,B\subset V_{\overrightarrow{G}}$,
we   define $\overline{\phi}_{\overrightarrow{G}}(A,B)$ by
	\begin{equation}\label{eqn:dirphi}
		\overline{\phi}_{\overrightarrow{G}}(A,B) \triangleq \frac{2{w}_{\overrightarrow{G}}(A,B)}{\vol_{\out}(A) + \vol_{\ino}(B)},
	\end{equation}
 where $${w}_{\overrightarrow{G}}(A,B)\triangleq \sum_{\substack{ (u,  v) \in E \\ u\in A, v\in B}} w(u,v)$$ is the sum of the weights of the edges from $A$ to $B$. 
For every $k \in \mathbb{N}$, the $k$-way directed dual Cheeger constant is defined by
	\begin{equation}
		\bar{\rho}_{\overrightarrow{G}}(k) \triangleq \max_{(A_1,B_1),\ldots,(A_k,B_k)} \min_{1 \le i \le k} \overline{\phi}_{\overrightarrow{G}} (A_i,B_i),
	\end{equation}
where the maximum is taken over	all the possible  $k$ pairs of subsets $(A_1,B_1),\ldots,(A_k,B_k)$ satisfying 	$	A_i \cap A_j=\emptyset, B_i \cap B_j=\emptyset, A_i \cap B_j = \emptyset	$
	for different $i,j\in[k]$, $A_i\cup B_i \ne \emptyset$ for any $i\in[k]$.  By definition, a high value of   $\bar{\rho}_{\overrightarrow{G}}(k)$ implies that   $\overrightarrow{G}$ contains $k$ bipartite-like clusters $(A_1,B_1),\ldots,(A_k,B_k)$ such that   most  edges with their tails in $A_i$ have their head in $B_i$ and conversely most edges with their head in $B_i$ have their tail in $A_i$.  We prove that, when $\overrightarrow{G}$  presents a    structure of  $k$ bipartite-like clusters with respect to  $\bar{\rho}_{\overrightarrow{G}}(k)$, this structure  can be represented by a sparse    graph $\overrightarrow{G^*}$ with $\widetilde{O}(n)$ edges, and   $\overrightarrow{G^*}$ can be constructed online in nearly-linear time:

 \begin{theorem}[Result for directed graphs]\label{thm:digraph} 
 Let $\overrightarrow{G} = (V_{\overrightarrow{G}},E_{\overrightarrow{G}},w_{\overrightarrow{G}})$ be a directed and weighted graph of $m$ edges, and assume that $\overrightarrow{G}$ contains $k$ directly bipartite-like   clusters $(A_1,B_1), \ldots, (A_k, B_k)$  with respect to  $\bar{\rho}_{\overrightarrow{G}}(k)$. Then, there is an algorithm that runs in  $\widetilde{O}(m)$ time and 
 computes a sparsifier  $\overrightarrow{G^*} = (V_{\overrightarrow{G}},F \subset E_{\overrightarrow{G}}, \widetilde{w})$,  such that      these $k$ directed bipartite-like clusters   of $\overrightarrow{G}$ are preserved in $\overrightarrow{G^*}$ with high probability. That is, it holds with high probability that $\bar{\rho}_{\overrightarrow{G^*}}(k) = \Omega\left(\bar{\rho}_{\overrightarrow{G}}(k)\right)$, and $\overrightarrow{G^*}$ only contains $k$ directed bipartite-like clusters. 
	\end{theorem}

Now we examine the significance of Theorems~\ref{thm:undirected_main} and \ref{thm:digraph}. We first highlight that our algorithms preserve the cut values $w(A_i, B_i)$ between the pairs of vertex sets $A_i$ and $B_i$ for $1\leq i\leq k$; this objective is \emph{different} from the one for most  graph sparsification problems, which only preserve the cut values between vertex set $S$ and $V\setminus S$. Secondly, our algorithms preserve   $k$ bipartite-like clusters, and the value of $k$  in the output graph is the same as the   input graph. Thirdly, our second result works for  \emph{directed graphs}; this result is very interesting on its own since most sparsification algorithms are only applicable for undirected graphs. Finally, while the design of most graph sparsification algorithms are based on Laplacian solvers making it unpractical, our designed algorithms only use random sampling.

The design of our  algorithms is based on several new reductions and sampling routines, and our algorithms can be implemented online with the degree oracles.    As such one can run our algorithms online while exploring the underlying graph with existing local algorithms~(e.g., \citep{Andersen2010,LP13}), resulting in direct improvement on the  running time of the existing algorithms. 
To demonstrate this, 
we conduct experimental studies and show that our algorithms can be directly applied to significantly speed up the running time of the  the ones presented in \citep{MS21}, while preserving similar output results on both the synthetic and real-world datasets. 

\paragraph{Related Work.} Bipartite-like clusters are widely studied in both theoretical computer science and machine learning communities.   In theoretical computer science, 
\citet{Trevisan09} developed a spectral algorithm that finds a bipartite-like cluster  in an undirected graph, and used this to design an approximation algorithm for the max-cut problem. This result is improved by \citet{soto}.  \citet{Liu} studied the relationship between the $k$-way dual Cheeger constant and the eigenvalues of the normalised graph Laplacians, and developed a Cheer-type inequality.

In the  machine learning community,   bipartite-like clusters are employed to model   highly-correlated   data items of different types, and algorithms  finding these clusters are studied in different settings. 
\citet{Andersen2010}, \citet{LP13} and \citet{MS21} presented local algorithms that find bipartite-like clusters, and \citet{MS_hypergraph} presented an  algorithm that finds bipartite components in hypergraphs.  
\citet{CLSZ}  proved  that densely connected clusters in a  directed graph can be uncovered through spectral clustering on a complex-valued Hermitian matrix representation of directed graphs.
\citet{NP22} designed  a sublinear-time  oracle which, under a certain condition, correctly classified the membership of most vertices  in a    set of hidden planted ground-truth clusters in signed graphs.

 Our work  relates  to   finding clusters in \emph{disassortative} networks~\citep{mooreActiveLearningNode2011, PWC+2019, zhu2020beyond}, although most existing techniques are based on semi-supervised and global methods. 
 Our work is also  related to designing graph sparsification algorithms, e.g., \citep{ST11,BSS,CohenKPPRSV17,Lee017,LeeS18}. 
 We highlight that, while a spectral sparsifier   preserves the cut value $w(S, V\setminus S)$ between any  vertex set $S$ and its complement $V\setminus S$, our algorithms' output preserves the cut value $w(A_i, B_i)$ for pairs of vertex sets $A_i$ and $B_i$. Moreover, our algorithms are much easier to implement, and work for \emph{directed} graphs. 
 
\section{Preliminaries}

In this section we list the notation and preliminary results used in the analysis.

\paragraph{Matrix Representation of Graphs.} 
We always use $G = (V,E,w)$ to represent an  undirected and  weighted graph with $n$ vertices  and weight function $w:E \to \mathbb{R}_{\ge 0}$.  The degree of any vertex $u$ is defined as   $d_G(u) = \sum_{u \sim v} w(u,v)$, where the notation $u \sim v$ represents that   $u$ and $v$ are adjacent, i.e., $\{u,v\} \in E(G)$.  The normalised indicator vector of any $S \subset V$ is defined by \[\chi_S(v)= \sqrt{ \frac{d_G(v)}{ \vol_G(S)} }\] if $v\in S$, and $\chi_S(v) = 0$ otherwise. 
Let $A_G$ be the adjacency matrix of $G$ defined by $(A_G)_{u,v} = w(u,v)$ if $\{u,v\} \in E(G)$, and  $(A_G)_{u,v} = 0$ otherwise. The degree matrix $D_G$ of $G$ is a diagonal matrix defined by $(D_G)_{u,u} = d_G(u)$, and the normalised Laplacian of $G$ is defined by \[\mathcal{L}_G = I - D_G^{-1/2} A_G D_G^{-1/2}.\] We can also write the normalised Laplacian matrix  with respect to the indicator vectors of the vertices: for each vertex $v$, we define an indicator vector $\chi_v \in \mathbb{R}^n$ by $\chi_v(u) = \frac{1}{\sqrt{d_v}}$ if $u = v$, and $\chi_v(u) = 0$ otherwise. We further define $b_e = \chi_u - \chi_v$ for each edge $e = \{u,v\}$, where the orientation of $e$ is chosen arbitrarily. Then, we have \[\mathcal{L}_G = \ds\sum_{e = \{u,v\} \in E} w(u,v) \cdot b_eb_e^{\intercal}.\]
We also define
\[
\mathcal{J}_G \triangleq I + D_G^{-1/2} A_G D_G^{-1/2}.\]
For any symmetric matrix $A\in\mathbb{R}^{n\times n}$, let $ \lambda_1(A) \le \lambda_2(A) \le \cdots \le \lambda_n(A)$ be the eigenvalues of $A$. For ease of presentation, we always    use $0 = \lambda_1 \le \lambda_2 \le \cdots \le \lambda_n \le 2$ to express the eigenvalues of $\mathcal{L}_G$, with the corresponding orthonormal eigenvectors $f_1,f_2,\cdots,f_n$. With  slight abuse of notation, we use $\mathcal{L}_G^{-1}$ for the pseudo-inverse of $\mathcal{L}_G$, i.e., \[\mathcal{L}_G^{-1} \triangleq \sum_{i=2}^{n} \frac{1}{\lambda_i} f_i f_i^{\rot}.\] Note that when $G$ is connected, it holds that $\lambda_2>0$ and the matrix $\mathcal{L}_G^{-1}$ is well defined. We sometimes  drop the subscript $G$   when it is clear from the context. 

For any   $x \in \mathbb{R}^n$ we define   $\|x\| \triangleq \sqrt{\sum_{i=1}^n x_i^2}$, and for  any   $M \in \mathbb{R}^{n \times n}$ we define $$\|M\| = \max_{x \in \mathbb{R}^n \setminus \{\textbf{0}\}} \frac{\|Mx\|}{\|x\|}.$$

\paragraph{Graph expansion and Cheeger inequality.} 
For any undirected graph $G$, the expansion (or conductance) of any non-empty subset $S \subset V$ in $G$ is defined as \[\phi_G(S) \triangleq \frac{w_G(S,\bar{S})}{\vol_G(S)},\] where $\bar{S}$ is the complement of $S$.  We call subsets of
vertices   $S_1,S_2, \cdots, S_k$ a $k$-way partition of $G$ if $S_i \ne \emptyset$ for all $1 \le i \le k$, $S_i \cap S_j = \emptyset$ for $i \ne j$ and $\bigcup_{i=1}^k S_i = V$. For any $k \in \mathbb{N}$, the $k$-way expansion constant is defined as 
\[
\rho_G(k) = \min_{S_1,S_2, \cdots, S_k} \max_{1 \le i \le k} \phi_G (S_i),\]
where the minimum is taken over all possible $k$-way partitions of $G$.  %The classical Cheeger inequality states that 	\[ 	\lambda_2/2\leq \rho_G(2)\leq \sqrt{2\lambda_2}; 	\] 
\citet{higher-order-cheeger}  proves 
the following higher-order Cheeger inequality: 
\begin{lem}[Higher-order Cheeger Inequality, \citep{higher-order-cheeger}]
	It holds for any undirected graph $G$ of $n$ vertices and integer $1\leq k \leq n$ that 
	\[  \lambda_k/2 \le \rho_G(k) \le C k^2 \sqrt{\lambda_k},\]
	where $C$ is a universal constant.
\end{lem}

Generalising this,  \citet{Liu}    proves the following higher-order dual-Cheeger inequality: 
\begin{lem}[Higher-order dual-Cheeger Inequality, \citep{Liu}]\label{lem:dualCheeger}
	It holds for any undirected graph $G$ of $n$ vertices and integer   $1 \le k \le n$ that 
	\[(2-\lambda_{n-k+1})/2  \le 1- \bar{\rho}_G(k) \le Ck^3 \sqrt{2-\lambda_{n-k+1}},\]
	%%	$$ or in another form, $$ 		\frac{(1-\bar{\rho}(k))^2}{C^2k^6} \le 2-\lambda_{n-k+1} \le 2(1-\bar{\rho}(k)), 		$$ 
	where $C$ is a universal constant.
\end{lem}

The higher-order dual Cheeger inequality can be viewed as a quantitative version of the fact that $\lambda_{n-k+1} = 2$ if and only if $G$ has at least $k$ bipartite connected components.

\section{Proof of Theorem~\ref{thm:undirected_main}  \label{sec:undirected}}

In this section we present a nearly-linear time  sparsification algorithm such that every bipartite-like cluster  in an undirected graph $G$ is approximately preserved in the sparsifed graph $G^*$, and sketch the proof. Our result is as follows:

	\begin{theorem}[Formal Statement of Theorem~\ref{thm:undirected_main}]\label{thm:undirected}
		There exists a nearly-linear time algorithm that, given an input graph $G = (V,E,w)$ with $\bar{\rho}_G(k)\geq \frac{1}{\log n}$ for constant  some $k$, 
		with high probability computes a sparsifier $G^* = (V,F \subset E, \widetilde{w})$ with $|F| = O \left( \frac{n\cdot  \log^3 n}{2 - \lambda_{n-k}}\right)$ edges such that the following hold: (1)   $\bar{\rho}_{G^*}(k) = \Omega(\bar{\rho}_G(k))$; (2)  $\lambda_{k+1}(\mathcal{J}_{G^*}) = \Theta(\lambda_{k+1}(\mathcal{J}_G))$.  
	\end{theorem}
    
The first statement of Theorem~\ref{thm:undirected} shows that the $k$  bipartite-like clusters   of $G$ is approximately preserved in $G^*$, and together with Lemma~\ref{lem:dualCheeger} the second statement   shows that the number of bipartite-like clusters   in $G$ and $G^*$ is the same.

\paragraph{Algorithm.}
Our algorithm  is similar with \citep{SZ19} at a high level,  and is based on sampling edges in $G$ with carefully defined probabilities. Formally, for an input undirected graph $G=(V,E,w_G)$, the algorithm starts with $G^*=(V,\emptyset, \widetilde{w})$ and samples every edge $u\sim v$ in $G$ with probability
$p_e \triangleq p_u(v)+p_v(u)-p_u(v) \cdot p_v(u),$
where
	\begin{align}
		\lefteqn{p_u(v)}\nonumber\\
       & \triangleq \min \left\{w_G(u,v) \cdot \frac{C \cdot  \log^3 n}{d_G(u) \cdot (2 - \lambda_{n-k})},1 \right\},\label{eqn:pu}
	\end{align}
for some constant $C$. For every sampled edge $e=\{u,v\}$, the algorithm adds $e$ to graph $G^*$, and sets $w_{G^*}(e) = w_G(e)/p_e$. Notice that, the choice of $C$ only changes the sampling probability by a constant factor, and   doesn't influence the asymptotic order of the sampled edges. Moreover, in practice we usually treat $\frac{C \cdot  \log^3 n}{2 - \lambda_{n-k}}$ as $O(\log ^ c n)$  for a constant $c$, and this only influences the   total number of sampled edges and the algorithm's   running time by a poly-logarithmic factor.

\paragraph{Proof Sketch of Theorem~\ref{thm:undirected}.}
We first  prove that the cut values between $A_i$ and $B_i$ in $G$ is preserved in $H$ for any $1\leq i\leq k$.  For any edge $e = \{u,v\}$, we define the  random variable $Y_e$ by  $Y_e = w_G(u,v)/p_e$ with probability $p_e$, and $Y_e=0$ otherwise.
By defining $X = w_H(A_i,B_i)$, we prove that  $\E[X]=w_G(A_i,B_i)$ and
	\begin{align*}		\lefteqn{\E\left[X^2\right]}\\
    &\le \frac{2- \lambda_{n-k}}{C \cdot  \log^3 n} \ds \sum_{\substack{e = \{u,v\} \\ u \in A_i, v \in B_i}} w(u,v) \cdot \left(\frac{d_G(u)+d_G(v)}{2}\right).
  \end{align*}
	Let $\{(A_i,B_i)\}_{i=1}^k$ be the optimal clusters corresponding to   $\bar{\rho}(k)$. 
   	Then, we have  for every $1 \le i \le k$ that  $$\bar{\rho}_G(k) \le \overline{\phi}_G (A_i,B_i) = \frac{2 w_G(A_i,B_i)}{\vol_G(A_i \cup B_i)},$$ which implies
\[
		\frac{\bar{\rho}_G(k)}{2} \cdot \vol_G(A_i \cup B_i) \le \sum_{\substack{e = \{u,v\} \\ u \in A_i, v \in B_i}} w_G(u,v).
  \]
	Applying the  Chebyshev's inequality,  we have for any constant $c\in\mathbb{R}^+$ that 
	\begin{align}\nonumber
		\begin{split}
		& \textbf{P}\left[|X-\E[X]| \ge c \cdot \E[X]\right]
   \le \frac{\E[X^2]}{c^2 \cdot \E[X]^2}\\
		%	&\le \frac{\frac{2- \lambda_{n-k}}{C \cdot  \log^3 n} \left( \sum_{\substack{e = \{u,v\} \\ u \in A_i, v \in B_i}} w_G(u,v) \cdot \left(\frac{d_G(u)+d_G(v)}{2}\right) \right)}{c^2 \cdot \left( \frac{\bar{\rho}_G(k)}{2} \cdot \vol_G(A_i \cup B_i)\right)^2}\\
			&\le \frac{2\cdot (2- \lambda_{n-k})}{c^2 \cdot C\cdot  \log^3 n \cdot \bar{\rho}_G(k)^2}\\
            & \cdot \frac{\left(\max_{\substack{e = \{u,v\} \\ u \in A_i, v \in B_i}} \{d_G(u)+d_G(v)\}\right)}{\vol_G(A_i \cup B_i)^2} \cdot \sum_{\substack{e = \{u,v\} \\ u \in A_i, v \in B_i}} w_G(u,v).
		\end{split}
	\end{align}
	Since $\vol_G(A_i \cup B_i) = \sum_{u \in A_i} d_G(u) + \sum_{v \in B_i} d_G(v)$ and $d_G(u) = \sum_{u \sim v} w_G(u,v)$, we have
    \begin{align*}
        &\max_{\substack{e = \{u,v\} \\ u \in A_i, v \in B_i}} \left(d_G(u)+d_G(v)\right)\\
        & \le \sum_{u \in A_i} d_G(u) + \sum_{v \in B_i} d_G(v) = \vol_G(A_i \cup B_i)
    \end{align*} and $
	\sum_{\substack{e = \{u,v\} \\ u \in A_i, v \in B_i}} w_G(u,v) \le \vol_G(A_i \cup B_i)$.
 Applying these gives us that
 \begin{align*}
     & \textbf{P}\left[|X-\E[X]| \ge c \cdot \E[X]\right] \\
     & \le \frac{2(2- \lambda_{n-k})}{c^2 \cdot C \cdot  \log^3 n \cdot \bar{\rho}(k)^2} = O\left(\frac{1}{ \log n}\right).
 \end{align*}
	Hence, by the union bound, we have that $
		w_H(A_i,B_i) = \Omega\left(w_G(A_i,B_i)\right) \text{ for all } 1\le i \le k$.
The proof of the second statement of Theorem~\ref{thm:undirected} can be found in the appendix. Finally, the total number of edges in $H$ follows by the definition of sampling probability and the Markov inequality.  
This completes the  proof of Theorem~\ref{thm:undirected}.

\section{Proof of Theorem~\ref{thm:digraph}\label{sec:directed}}

In this section we present a nearly-linear time  sparsification algorithm such that 
every directed bipartite-like cluster  in a directed  graph   is approximately preserved in the output sparsifier, and  prove Theorem~\ref{thm:digraph}. Specifically, for a digraph $\overrightarrow{G}$ that contains exactly  $k$ pairs of    $(A_1,B_1),\ldots,(A_k,B_k)$ with high values of  $\overline{\phi}_{\overrightarrow{G}}(A_i,B_i)$
 for every $1\leq i\leq k$, our objective is to construct a sparse digraph $\overrightarrow{G^*}$, such that (i) the values of $\overline{\phi}_{\overrightarrow{G^*}} (A_i,B_i)$ are high for every $1\leq i\leq k$ and  (ii) the number of such pairs in $\overrightarrow{G^*}$ is the same as $\overrightarrow{G}$.  Our result is as follows:

 \begin{theorem}[Formal Statement of Theorem~\ref{thm:digraph}]\label{thm:digraph_formal} There is a nearly-linear time algorithm that, given a directed and weighted graph $\overrightarrow{G} = (V_{\overrightarrow{G}},E_{\overrightarrow{G}},w_{\overrightarrow{G}})$ with $n$ vertices and $k$ directed bipartite-like clusters satisfying  $\bar{\rho}_{\overrightarrow{G}}(k)=1-o(1/k)$
 as input, with high probability computes a sparsifier $\overrightarrow{G^*} = (V_{\overrightarrow{G}},F \subset E_{\overrightarrow{G}}, \widetilde{w})$ such that $\bar{\rho}_{\overrightarrow{G^*}}(k) = \Omega\left(\bar{\rho}_{\overrightarrow{G}}(k)\right)$. Moreover, the total number of edges in the output graph is nearly-linear in 
 $n$.
\end{theorem}

Before sketching our technique,   recall that, for undirected graphs,     the value of $k$ is proven to be identical for $G$ and $G^*$ by analysing the eigenvalues of $\mathcal{J}_G$ and $\mathcal{J}_{G^*}$ and applying the higher-order dual-Cheeger inequality~(Lemma~\ref{lem:dualCheeger}). However, a natural matrix representation for directed graphs could result in complex-valued eigenvalues, and there is no analogue of Lemma~\ref{lem:dualCheeger} for directed graphs. To overcome this, our developed algorithm is based on a novel reduction from a directed graph to an undirected one, and its reverse operation. Specifically, 
 our designed algorithm consists of the following three steps:
 \begin{enumerate}
 \item for any input digraph $\overrightarrow{G}$, the algorithm constructs an undirected graph $H$ such that every directed bipartite-like cluster defined by  $(A_i, B_i)$ in $\overrightarrow{G}$ corresponds to a low-conductance set in $H$;
 \item the algorithm constructs a sparsifier $H^*$ of $H$, such  that   $H$ and $H^*$ have the same structure of clusters;
 \item the algorithm applies the sparsified undirected graph $H^*$ to construct a directed graph $\overrightarrow{G^*}$ of $\overrightarrow{G}$ that satisfies $\bar{\rho}_{\overrightarrow{G^*}}(k) = \Omega\left(\bar{\rho}_{\overrightarrow{G}}(k)\right)$.
 \end{enumerate}
See  Figure~\ref{fig:relationship} for illustration.
 
\begin{figure}[ht]
		\centering
		\begin{tikzpicture}[scale=0.7]
			\tikzset{edge/.style = {->,> = latex'}}
			\begin{scope}[execute at begin node=$, execute at end node=$]
				\node at (-0.5,0) {\left(\overrightarrow{G^*},\bar{\rho}_{\overrightarrow{G^*}}(k) \right)} ;
				\node at (6.5,0) {\left(H^*,\rho_{H^*}(k)\right)} ;
				\node at (-0.5,4) {\left(\overrightarrow{G},\bar{\rho}_{\overrightarrow{G}}(k) \right)} ;
				\node at (6.5,4) {\left(H,\rho_{H}(k)\right)} ;
			\end{scope}
			\draw[edge,thick] (4.5,0) -- (1.5,0) node[midway, above] {reverse} node[midway, below] {semi-double cover};
			\draw[edge,thick] (1.5,4) -- (4.5,4) node[midway, above] {semi-double} node[midway, below] {cover};
			\draw[edge,thick] (6,3.5) -- (6,0.5) node[near start, right] {graph} node[midway, right] {sparsification};
			\draw[edge,thick] (0,3.5) -- (0,0.5) node[near start, left] {graph} node[midway, left] {sparsification};
		\end{tikzpicture}
		\caption{A commutative diagram of our construction. To construct $\overrightarrow{G^*}$ from $\overrightarrow{G}$, we construct graphs $H$ and $H^*$ and prove the close relationships between $\overrightarrow{G}$, $H$, $H^*$, and $\overrightarrow{G^*}$. \label{fig:relationship}}
	\end{figure}
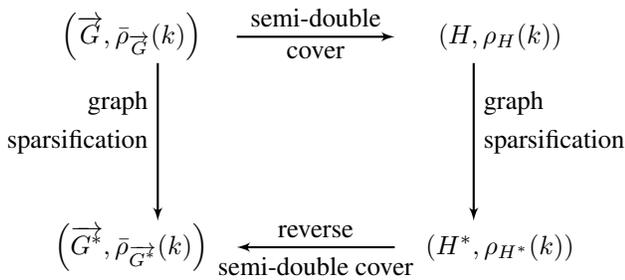

\paragraph{Constructing $H$ from $\overrightarrow{G}$.} Notice that,  to preserve    $\overline{\phi}_{\overrightarrow{G^*}} (A_i,B_i)$, the cut values  $w(A_i, B_i)$ between $A_i$ and $B_i$ need to be approximately preserved in a sparsified directed graph; this objective is  different from the most graph sparsification one, which only preserves the cut value between any set $S$ and its complement. To overcome this,      we construct an undirected graph $H$
such that every bipartite-like cluster defined by  $(A_i, B_i)$ in $\overrightarrow{G}$ corresponds to a low-conductance set in $H$. Specifically, for a weighted 
digraph ${\overrightarrow{G}} = (V_{\overrightarrow{G}},E_{\overrightarrow{G}},w_{\overrightarrow{G}})$, we   construct its semi-double cover $H = (V_H,E_H,w_H)$ as follows: (1) every vertex $v \in V_{\overrightarrow{G}}$ has two corresponding vertices $v_1,v_2 \in V_H$; (2)  for every edge $(u,v) \in E_{\overrightarrow{G}}$, we add the edge $\{u_1,v_2\}$ in $E_H$. 
See Figure~\ref{fig:sdc} for illustration.
	\begin{figure}[ht]
		\centering
		\begin{tikzpicture}[scale=0.6]
			\Vertex[x=0,y=0,opacity =.2,label=$c$]{c}
			\Vertex[x=0,y=2,opacity =.2,label=$a$]{a}
			\Vertex[x=2,y=0,opacity =.2,label=$d$]{d}
			\Vertex[x=2,y=2,opacity =.2,label=$b$]{b}
			
			\Edge[Direct](a)(b)
			\Edge[Direct](a)(c)
			\Edge[Direct](a)(d)
			\Edge[Direct](c)(d)
			\Edge[Direct](b)(c)
			
			\draw[arrows = {-Stealth[length=10pt, inset=5pt]}] (3.2,1)   -- (4.2,1);
		\end{tikzpicture}
		\hspace{1cm}
		\begin{tikzpicture}[scale=0.6]
			\Vertex[x=0,y=0,opacity =.2,label=$a_2$]{a2}
			\Vertex[x=2,y=0,opacity =.2,label=$b_2$]{b2}
			\Vertex[x=4,y=0,opacity =.2,label=$c_2$]{c2}
			\Vertex[x=6,y=0,opacity =.2,label=$d_2$]{d2}
			
			\Vertex[x=0,y=2,opacity =.2,label=$a_1$]{a1}
			\Vertex[x=2,y=2,opacity =.2,label=$b_1$]{b1}
			\Vertex[x=4,y=2,opacity =.2,label=$c_1$]{c1}
			\Vertex[x=6,y=2,opacity =.2,label=$d_1$]{d1}
			
			\Edge(a1)(b2)
			\Edge(a1)(c2)
			\Edge(a1)(d2)
			\Edge(c1)(d2)
			\Edge(b1)(c2)
		\end{tikzpicture}
		\caption{Illustration of the semi-double cover construction. A directed graph  of $n$ vertices~(left) corresponds to an undirected and bipartite graph of $2n$ vertices~(right). }
		\label{fig:sdc}
	\end{figure}
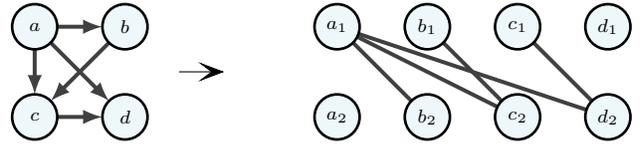

Next we analyse the properties of the reduced graph. 
 Let ${\overrightarrow{G}}$ be a directed graph with semi-double cover $H$. 
	For any $S \subset V_{\overrightarrow{G}}$, we define $S_1 \subset V_H$ and $S_2 \subset V_H$ by $S_1 \triangleq \{v_1 | v \in S\}$ and $S_2 \triangleq \{v_2 | v \in S\}$. A subset $S$ of $V_H$ is called \emph{simple} if $|\{v_1,v_2\} \cap S| \le 1$ holds for all $v \in V_{\overrightarrow{G}}$. The following lemma  develops  a   relationship between the \emph{flow ratio} from $A$ to $B$ defined by
 \begin{equation}\label{eqn:flow}
		f_{\overrightarrow{G}}(A,B) \triangleq 1 - \overline{\phi}_{\overrightarrow{G}}(A,B)
	\end{equation}
 and $\Phi_H(A_1 \cup B_2)$, for any $A,B$.

	\begin{lem}\label{lem:Fphi}
		Let ${\overrightarrow{G}}$ be a directed graph with semi-double cover $H$. Then, it holds for any $A,B \subset V_{\overrightarrow{G}}$ that $f_{\overrightarrow{G}}(A,B) = \phi_H(A_1 \cup B_2)$. Similarly, for any simple set $S \subset V_H$, let $A = \{u: u_1 \in S\}$ and $B = \{u: u_2 \in S\}$. Then, it holds that $f_{\overrightarrow{G}}(A,B) = \phi_H(S)$.
	\end{lem}

Lemma~\ref{lem:Fphi} proves a one-to-one correspondence between any bipartite-like cluster in  $\overrightarrow{G}$ and a vertex set in $H$. Building on this, we prove that this one-to-one correspondence can be  generalised between any $k$ bipartite-like clusters  in $\overrightarrow{G}$ and $k$ disjoint vertex sets in $H$. Moreover,  the structure of $k$ bipartite-like clusters   in $\overrightarrow{G}$ is preserved by a collection of $k$ disjoint vertex sets of low conductance in $H$. 
 
\begin{lem}\label{lem:directed_reduction}
For any  directed and weighted graph $\overrightarrow{G} = (V_{\overrightarrow{G}},E_{\overrightarrow{G}},w_{\overrightarrow{G}})$ and $k\in\mathbb{N}$, it holds that \begin{equation}\label{eqn:rhoG-rhoH}
		\bar{\rho}_{\overrightarrow{G}}(k) = 1 -  \min_{C_1,\ldots,C_k} \max_{1 \le i \le k} \phi_H(C_i),
	\end{equation}
	where the minimum is taken over  $k$ disjoint simple subsets of $V_H$ defined by $C_i = A_{i_1} \cup B_{i_2}$ for $1 \le i \le k$. 
\end{lem}

\paragraph{Sparsification of $H$.} Next we   construct a sparse representation of $H$, denoted by $H^*$, such that the $k$ vertex sets of low conductance is preserved in $H^*$. To achieve this, we apply the following result to construct a cluster-preserving sparsifier.

\begin{lem}[\citep{SZ19}]\label{lem:spar}		There exists a nearly-linear time algorithm that, given a graph $G = (V,E,w)$ with $k$ clusters as input, with probability at least $9/10$, computes a sparsifier $H = (V,F \subset E, \widetilde{w})$ with $|F| = O((1/\lambda_{k+1}) \cdot n \log n)$ edges such that the following holds:  (1) 
  it holds for any $1 \le i \le k$ that $\phi_H(S_i) = O(k \cdot \phi_G(S_i))$, where $S_1,\cdots,S_k$ are the optimal clusters in $G$ that achieves $\rho(k)$; 			(2) $\lambda_{k+1}(\mathcal{L}_H) =\Omega(\lambda_{k+1}(\mathcal{L}_G))$. 		  \end{lem}  

\paragraph{Constructing $\overrightarrow{G^*}$ from $H^*$.} 
Finally, we  construct a directed graph $\overrightarrow{G^*}$ from $H^*$ such that the original $k$ directed bipartite-like clusters   in $\overrightarrow{G}$ is preserved in $\overrightarrow{G^*}$. To achieve this, we introduce the following \emph{reverse semi-double cover}:
\begin{defn}[reverse semi-double cover]
Given any double cover graph $H^*=(V_{H^*}, E_{H^*}, w_{H^*})$ as input, the reverse semi-double cover of $H^*$ is a directed graph $\overrightarrow{G^*}=(V_{\overrightarrow{G^*}}, E_{\overrightarrow{G^*}}, w_{\overrightarrow{G^*}})$ constructed as follows:
\begin{itemize}
    \item every pair of vertices $u_1$ and $u_2$ in $V_{H^*}$ corresponds to a vertex $v\in V_{\overrightarrow{G^*}}$;
    \item we add  an edge $(u, v)$ to $E_{\overrightarrow{G}}$ if there is edge $\{u_1, v_2\}\in E_{H^*}$, and set $w_{\overrightarrow{G^*}}(u,v)= w_{H^*}(u_1, v_2)$.
\end{itemize}
\end{defn}
One   might think that the reverse double cover plays an exact opposite role of the double cover, however  it is  not the case.
In particular, while our constructed  subsets $C_1,\ldots, C_k$ in the first step  are always simple   in $H$~(cf.~Lemma~\ref{lem:directed_reduction}),  the $k$ subsets corresponding to $\rho_H(k)$ are not necessarily   simple. As a result, 
	\[
	\min_{C_1,\ldots,C_k} \max_{1 \le i \le k} \phi_H(C_i) = \rho_H(k)
	\]
 doesn't hold in general, and there is no direct correspondence between $C_1,\ldots, C_k$ in $H$ and the $k$ directed 
 bipartite-like clusters  in $\overrightarrow{G^*}$ that correspond to $\bar{\rho}_{\overrightarrow{G^*}}(k)$.

To analyse  $\rho_{\overrightarrow{{G^*}}}(k)$, for any set $S \subset V_H$ we   partition the set into two subsets $S_1$ and $S_2$ defined by $S_1 = S \cap (A_{i_1} \cup B_{i_2})$ and $S_2 = S \cap (A_{i_2} \cup B_{i_1})$.
    For example, following Figure~\ref{fig:sdc}, if   $A_i = \{a,c\}$ and  $B_i = \{b,d\}$ and the set $S \subset V_H$ is $S = \{a_1,b_1,b_2,c_1,c_2\}$, then we have $S_1 = \{a_1,b_2,c_1\}$ and $S_2 = \{b_1,c_2\}$. 
    As $A_i$ and $B_i$ are densely connected in $H$, there are   few edges within   $A_i$ and $B_i$ for $1 \le i \le k$. 
	Hence,  there are very few edges between $S_1$ and $S_2$ for any  $S \subset V_H$. Without loss of generality, we   assume that 
	\[
	\frac{2w_H(S_1,S_2)}{w_H(S_1,\bar{S_1})+w_H(S_2,\bar{S_2})} \le c
	\]
 for some constant $c<1$. 
	Simplifying the inequality above  we get \begin{align*}
	& w_H(S_1,\bar{S_1})+w_H(S_2,\bar{S_2})-2w_H(S_1,S_2)\\
    & \ge (1-c)\cdot \left[w_H(S_1,\bar{S_1})+w_H(S_2,\bar{S_2})\right].
	\end{align*}
	Thus, for any not  necessarily simple vertex set   $S \subset V_H$  we have
	\begin{align*}
		\phi_H(S) &= \frac{w_H(S,\bar{S})}{\vol(S)}\\
	%	&= \frac{w_H(S_1,\bar{S})+w_H(S_2,\bar{S})}{\vol(S_1)+\vol(S_2)}\\
	%	&= \frac{w_H(S_1,\bar{S_1})-w_H(S_1,S_2)+w_H(S_2,\bar{S})-w_H(S_1,S_2)}{\vol(S_1)+\vol(S_2)}\\
		& = \frac{w_H(S_1,\bar{S_1})+w_H(S_2,\bar{S})-2w_H(S_1,S_2)}{\vol(S_1)+\vol(S_2)}\\
	%	&\ge (1-c)\cdot  \frac{w_H(S_1,\bar{S_1})+w_H(S_2,\bar{S_2})}{\vol(S_1)+\vol(S_2)}\\
		&\ge (1-c) \cdot \min \left\{\frac{w_H(S_1,\bar{S_1})}{\vol(S_1)} , \frac{w_H(S_2,\bar{S_2})}{\vol(S_2)}\right\}\\
        &= (1-c) \cdot \min \left\{\phi_H(S_1),\phi_H(S_2)\right\},
	\end{align*}
	where the last inequality follows by the median inequality. Thus, for every   $S \subset V_H$, there is  a simple set $T \subset V_H$ such that $\phi_H(S) \ge (1-c) \cdot \phi_H(T)$. Moreover, for any collection of $k$-disjoint sets $S_1,S_2,\cdots,S_k$, where $S_i \subset V_H$ we have a collection of $k$-disjoint simple sets $T_1,T_2,\cdots,T_k$, where $T_i \subset V_H$, such that 	\[
	\max_{1 \le i \le k} \phi_H(S_i) \ge (1-c) \cdot \max_{1 \le i \le k} \phi_H(T_i). \]
	Taking minimum over all such collection of $k$-disjoint subsets of $V_H$ gives us that 
	\begin{align*}
		&\min_{S_1,S_2,\cdots,S_k} \max_{1 \le i \le k} \phi_H(S_i) \\
        &= \rho_H(k) \ge (1-c) \cdot \min_{T_1,T_2,\cdots,T_k} \max_{1 \le i \le k} \phi_H(T_i),
	\end{align*}
	where in the second half of the inequality the minimum is taken over collection of $k$-disjoint simple subsets of $V_H$. On one hand, rearranging the above inequality we have
	\begin{equation}\label{eqn:rhoH1}
		\frac{1}{1-c} \cdot	\rho_H(k) \ge  \min_{T_1,T_2,\cdots,T_k} \max_{1 \le i \le k} \phi_H(T_i),
	\end{equation}
	and on the other hand, since the collection of $k$-disjoint simple subsets of $V_H$ is a sub-collection of the collection of $k$-disjoint subsets of $V_H$, we have 
	\begin{equation}\label{eqn:rhoH2}
		\min_{T_1,T_2,\cdots,T_k} \max_{1 \le i \le k} \phi_H(T_i) \ge \rho_H(k).
	\end{equation}
	Thus, combining \eqref{eqn:rhoH1} and \eqref{eqn:rhoH2}, we have 
	\begin{equation}\label{eqn:rhoH3}
		\frac{1}{1-c} \cdot	\rho_H(k) \ge \min_{T_1,T_2,\cdots,T_k} \max_{1 \le i \le k} \phi_H(T_i) \ge \rho_H(k).
	\end{equation}
	Further, combining \eqref{eqn:rhoG-rhoH} and \eqref{eqn:rhoH3} we have 
	\begin{equation}\label{eqn:bounds}
		1 - \frac{1}{1-c} \cdot \rho_H(k) \le \bar{\rho}_{\overrightarrow{G}}(k) \le 1- \rho_H(k).
	\end{equation}

\begin{figure*}[t]
    \centering
    \begin{subfigure}{.42\textwidth}
    \centering
    \includegraphics[width=.95\linewidth]{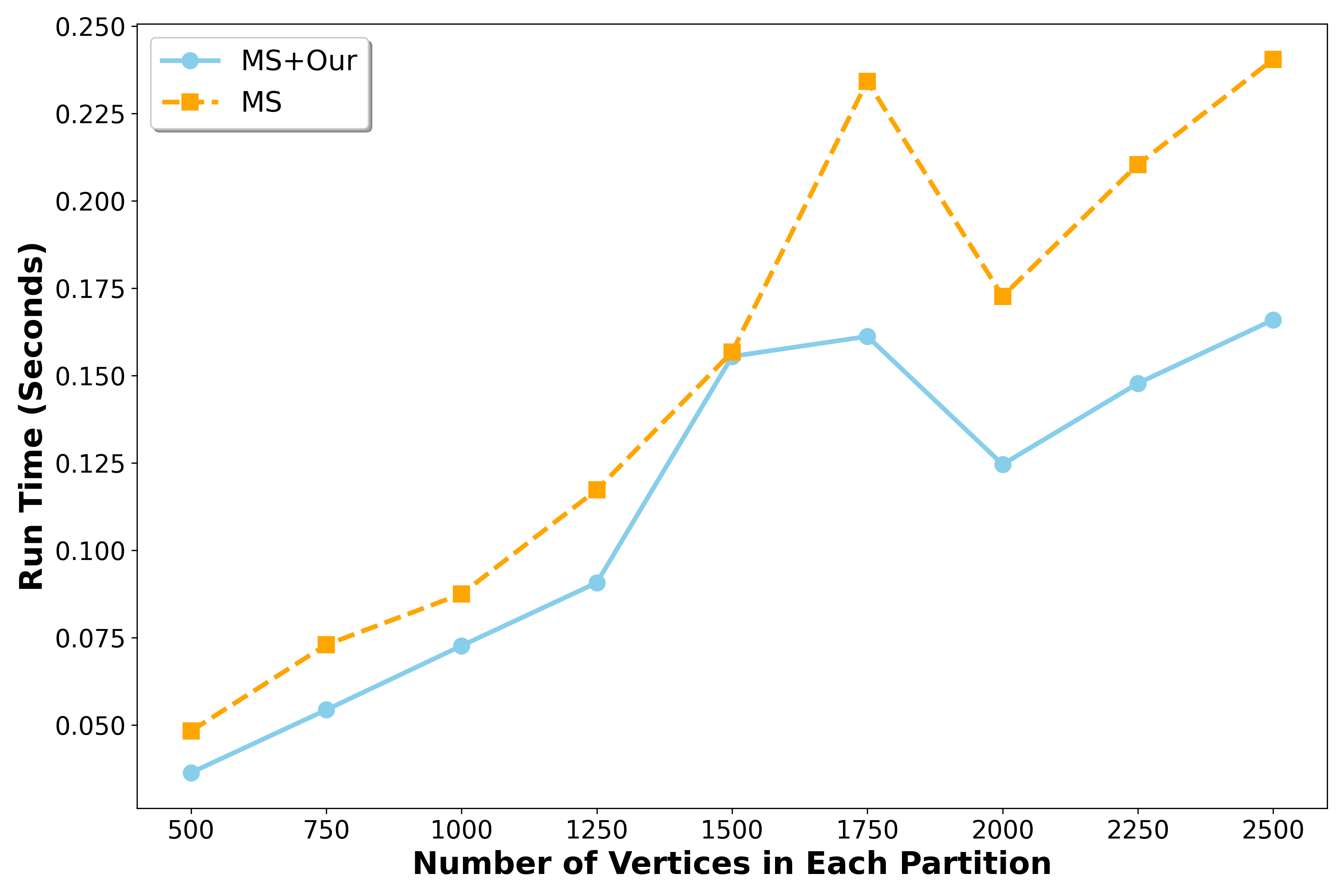} 
    \caption{Runtime comparison}
    \label{SUBFIGURE LABEL 1}
    \end{subfigure}
    \begin{subfigure}{.42\textwidth}
    \centering
    \includegraphics[width=.95\linewidth]{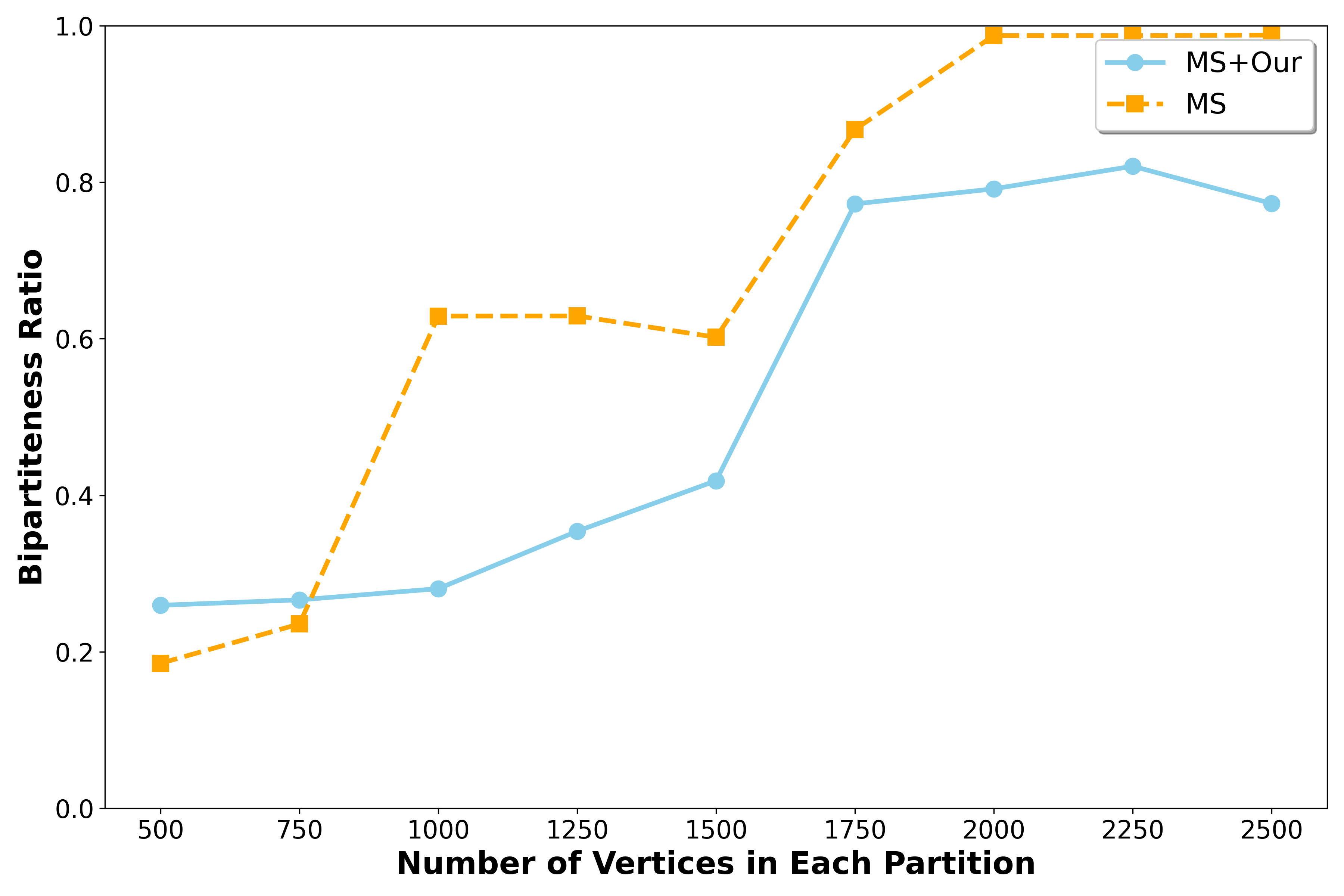}  
    \caption{Bipartiteness Ratio comparison}
    \label{SUBFIGURE LABEL 2}
    \end{subfigure}
    \caption{Runtime and bipartiteness comparison between \textsf{MS} and our algorithm by fixing $p = 0.3$, $q = 0.1p$ and varying the number of vertices  between $500$ and $2,500$ in each partition. }
\label{fig:exp_result1}
\end{figure*}

\paragraph{Proof of Theorem~\ref{thm:digraph}.} Now we are ready to prove Theorem~\ref{thm:digraph}. 
	Since $\overrightarrow{G}$ is a directed graph with $k$ bipartite-like clusters, the value of $\bar{\rho}_{\overrightarrow{G}}(k)$ is high;
  together with \eqref{eqn:bounds}, this implies that $\rho_H(k) = o(1)$. By  Lemma~\ref{lem:spar}, we know that there exists a sparsifier $H^*$ of $H$,   such that $\rho_{{H^*}}(k) = O(k \cdot \rho_{H}(k))$. Thus, we can conclude that $\rho_{H^*}(k) = o(1)$. Hence, applying \eqref{eqn:bounds} for $\overrightarrow{{G^*}}$ and $H^*$ we have
	\begin{equation}\label{eqn:bounds1}
		1 - \frac{1}{1-c} \cdot \rho_{H^*}(k) \le \bar{\rho}_{\overrightarrow{G^*}}(k) \le 1- \rho_{H^*}(k).
	\end{equation}
	Finally, using the fact that $\rho_{H^*}(k) = o(1)$, we   conclude that $\bar{\rho}_{\overrightarrow{G^*}}(k)$ is  close to $1$ and hence the structure of $\overrightarrow{G}$ will be preserved in $\overrightarrow{G^*}$. Moreover, by the construction of $H$, and $H^*$, and $\overrightarrow{G^*}$, the value of $k$ is preserved. 

 For the running time, notice that all the intermediate graphs $H$ and $H^*$ can be constructed locally, and   it's sufficient to examine every edge of the input graph $\overrightarrow{G}$ once throughout the execution of the algorithm. This implies the nearly-linear running time of our overall algorithm.  
 Combining everything above above proves Theorem~\ref{thm:digraph}.

\section{Experiments \label{sec:experiment}}	
We evaluate the performance of our proposed algorithms on synthetic  and real-world datasets. We employ the algorithms presented in \citep{MS21} as the baseline algorithms, and examine the speedup of their algorithms when applying our sparsification algorithms as subroutines. 
Notice that, as all the involved operations of our algorithms can be performed locally, one can run our graph sparsification algorithms online while exploring the underlying graph with a local algorithm. For ease of presentation, in this section we call the  local algorithm in \citep{MS21} with our sparsification framework our algorithm. 
%For undirected graphs, we compare the performance of our algorithm against the previous existing algorithm in \cite{Peter}, through the synthetic dataset to demonstrate the significance of our algorithm. 	For directed graphs, we compare the performance of our algorithm with the existing local algorithm in \citet{MS21}. 
All experiments were performed on a HP ZBook Studio with 11th Gen Intel(R) Core(TM) i7-11800H @ 2.30GHz processor and 32 GB of RAM. Our code can be downloaded
from \href{https://github.com/suranjande4/Online-Sparsification-of-Bipartite-Like-Clusters-in-Graphs}{https://github.com/suranjande4/Online-Sparsification-of-Bipartite-Like-Clusters-in-Graphs}.
\subsection{Results for Undirected Graphs}
%Now we evaluate the performance of our algorithm for undirected graphs on synthetic and real-world datasets.

\noindent \textbf{Synthetic Dataset.} We compare the performance of our algorithm with the   \textsc{LocBipartDC} algorithm presented in \cite{MS21}, which we refer to as \textsc{MS}, on synthetic graphs generated from the stochastic block model (SBM). Specifically, we assume that the graph has $k = 2$ clusters, say $C_1,C_2$, and the number of vertices in each cluster, denoted by $n_1$ and $n_2$ respectively, satisfies $n_1 = n_2$. Moreover, any pair of vertices $u \in C_i$ and $v \in C_j$ is connected with probability $p_{ij}$. We assume that $p_{12} = p_{21} = p$ and $p_{11} = p_{22} = q$, where $q = 0.1p$. Throughout the experiments, we leave the parameters $n$ and $p$ free but maintain the above relations. 

    Our algorithm sparsifies the underlying graph and   simultaneously applies the \textsc{MS} algorithm.
    We evaluate the quality of the output $(L,R)$ returned by each algorithm with respect to its bipartiteness ratio defined by $
    \beta(L,R) = 1 -  \overline{\phi}(L, R)$. 
    %Notice that a low $\beta(L, R)$ value means that there is a dense cut between $L$ and $R$, and there is a sparse cut between $L \cup R$ and $V \setminus (L \cup R)$. In particular, $\beta(L, R) = 0$ implies that $(L, R)$ forms a bipartite and connected component of $G$. 
    All our reported results are the average performance of each algorithm over $10$ runs, in which a random vertex from $C_1 \cup C_2$ is chosen as the starting vertex of the algorithm.   We generate graphs from the SBM such that $q = 0.1p$ and vary the size of the target set by varying $n_1$ between $500$ and $2,500$. In  Figure~\ref{fig:exp_result1}, we fix the probability $p = 0.3$ and vary the number of vertices $n_1 = n_2$ and compare both runtime and the bipartiteness ratio between the \textsc{MS} algorithm and our algorithm. One can observe that for a fixed probability $p$ as we increase the number of vertices, our algorithm takes much less time than the \textsc{MS} algorithm and maintains a similar bipartiteness ratio with the \textsc{MS} algorithm.
      %  \item We generate graphs from the SBM such that it has the  fixed number of vertices $n_1 = n_2 = 2,000$, vary $p$ between $0.4$ to $0.8$ and $q = 0.001p$. In the Figure~\ref{fig:exp_result12}, we fix the the number of vertices $n_1 = n_2$ and vary the probability $p$ and compare both runtime and the bipartiteness ratio between the \textsf{MS} algorithm and our algorithm.
  %  \end{itemize}

%  Note that following our sparsification procedure we apply the MS algorithm, which is a local algorithm and doesn't necessarily explore the entire graph; the MS algorithm terminates once it finds a bipartite-like cluster  with a certain guarantee. That's why it's sufficient for us to only examine the input graph with 2 clusters, and our experimental result is representative even for the graphs of $k\geq 3$ clusters.

\begin{table*}[ht]
\centering
\caption{Comparison of \textsc{MS} with our algorithm on the Militarised Interstate Disputes Dataset. We use the vertices corresponding to the listed countries in the first column as the seed vertex of the local algorithm.}
	\begin{tabular}{ccccc}
		\hline
		\textsc{Country Name}           & \textsc{MS}     & \textsc{MS}        & \textsc{Our Algo.} & 
        \textsc{Our Algo.}    \\  
		            &    \textsc{Runtime} & \textsc{Bipartiteness} & \textsc{Runtime} & \textsc{Bipartiteness} \\ \hline
		\textsc{USA}        
		            & 0.034  & 0.292      & 0.0044  & 0.285      \\ \hline
		\textsc{Netherlands}     
		            & 0.0351   & 0.307      & 0.0042   &  0.281     \\ \hline
		\textsc{Lithuania}      
		            & 0.0336   & 0.303       & 0.0043   & 0.165       \\ \hline
	\end{tabular}
    
  \label{tab:2}
\end{table*}

\begin{figure*}[ht]
\centering
\begin{subfigure}{.42\textwidth}
    \centering
    \includegraphics[width=.95\linewidth]{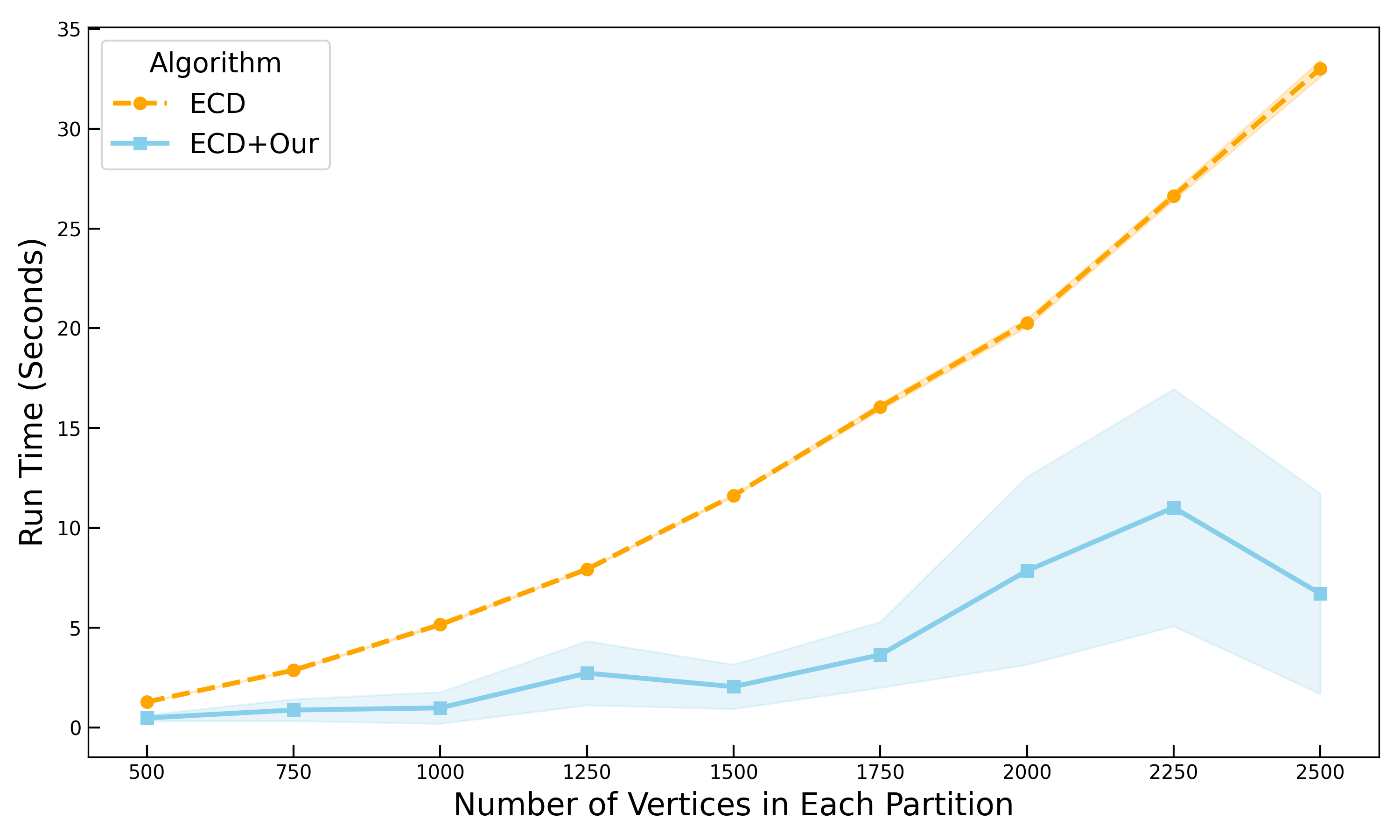}  
    \caption{Runtime comparison}
    \label{SUBFIGURE LABEL 11}
\end{subfigure}
\begin{subfigure}{.42\textwidth}
    \centering
    \includegraphics[width=.95\linewidth]{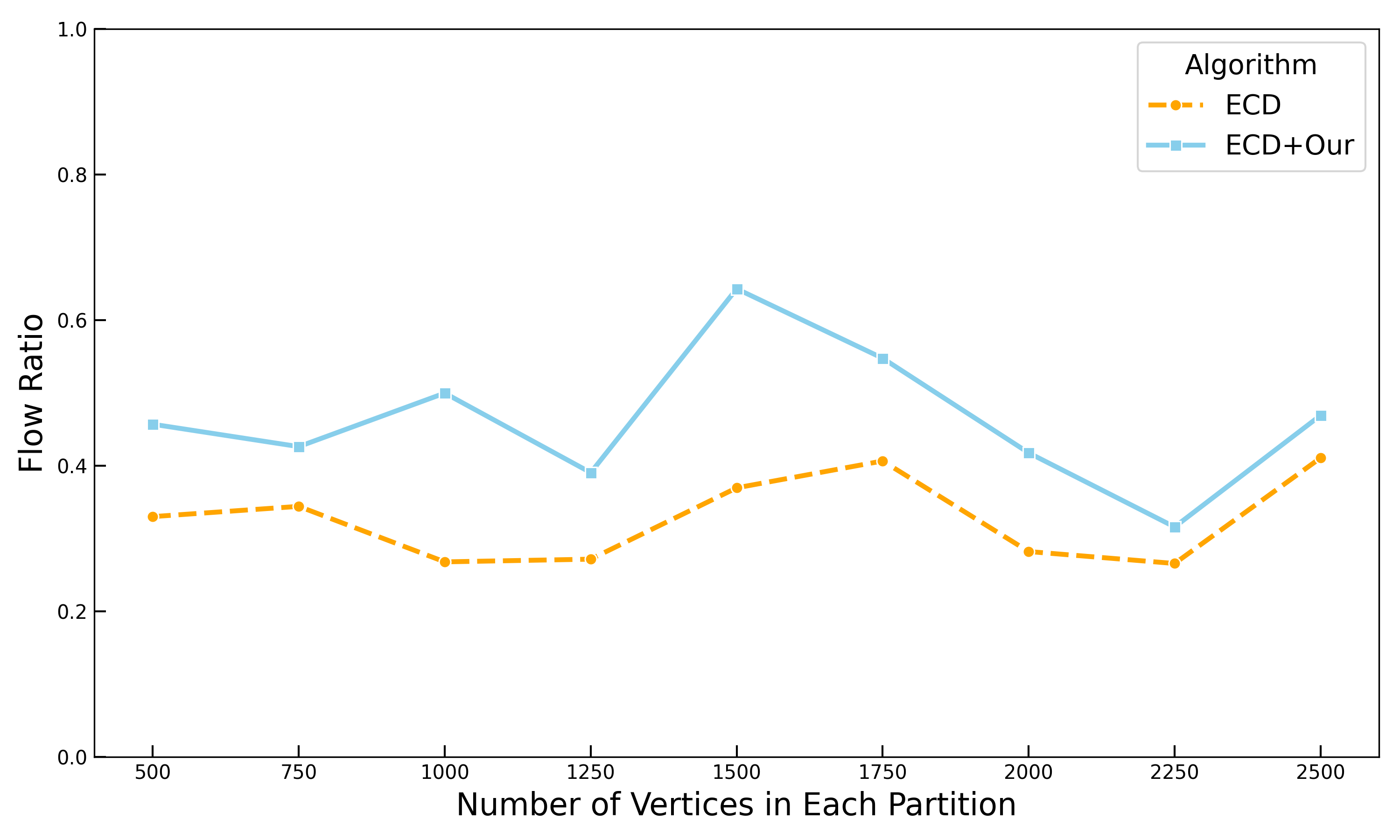}  
    \caption{Flow-ratio comparison}
    \label{SUBFIGURE LABEL 21}
\end{subfigure}
\caption{Runtime and flow-ratio comparison between \textsf{ECD} and our algorithm by fixing $\eta = 0.7$ and varying the number of vertices in each partition from $500$ to $2,500$.\label{fig:fig-4}} 
\end{figure*}

\noindent \textbf{Real-world Dataset.} We evaluate the performance of our algorithm on the Dyadic Militarised Interstate Disputes Dataset (v3.1)~\cite{midDataset}, which records every interstate dispute during 1900–1950, including the level of hostility resulting from the dispute and the number of casualties. We construct a graph from the dataset as follows: every country is represented by a vertex; two vertices are connected by an edge with weight $30$ if there is a war between the corresponding countries, and the two vertices are connected by an edge with weight $1$ if the corresponding countries have other dispute which is not part of an interstate war.  We set $\gamma=0.02$ for the \textsc{MS} algorithm, and Table~\ref{tab:2} compares the performance of the   \textsc{MS} algorithm with ours. This shows that   our algorithm takes much less time than the \textsc{MS} algorithm and maintains a similar bipartiteness ratio.   
	
%\begin{figure}[ht]   
%\centering     
%\begin{subfigure}{.48\textwidth}     
%\centering    
%\includegraphics[width=.95\linewidth]{Figures/Runtime-comp.png}      %\caption{Runtime comparison}     
%\label{SUBFIGURE LABEL 1}     
%\end{subfigure}     
%\begin{subfigure}{.48\textwidth}  
%\centering     
%\includegraphics[width=.95\linewidth]{Figures/Bpr-comp.png}       
%\caption{Bipartiteness Ratio comparison}     
%\label{SUBFIGURE LABEL 2}     
%\end{subfigure}    
%\caption{Runtime and bipartiteness comparison between \textsf{MS} and our algorithm by fixing $n_1 = n_2 = 2000$ and varying $p$ between $0.4$ to $0.8$ and $q = 0.001p$.}
%\label{fig:exp_result12}
%\end{figure}

% \he{Is Figure (a) correct? It looks that both algorithms' runtime are roughly the same} \joy{Yes, since we are fixing the number of vertices, so both runtime are roughly the same, as p increases, our performs better since the graph is dense. If we vary the number of vertices, then significant diff shows.}

\begin{figure*}[ht]
\centering
\begin{subfigure}{.32\textwidth}
    \centering
    \includegraphics[width=.95\linewidth]{Figures/runtime_eta_0.7-3.png}  
    \caption{$\eta = 0.7$}
    \label{SUBFIGURE LABEL 22}
\end{subfigure}
\begin{subfigure}{.32\textwidth}
    \centering
    \includegraphics[width=.95\linewidth]{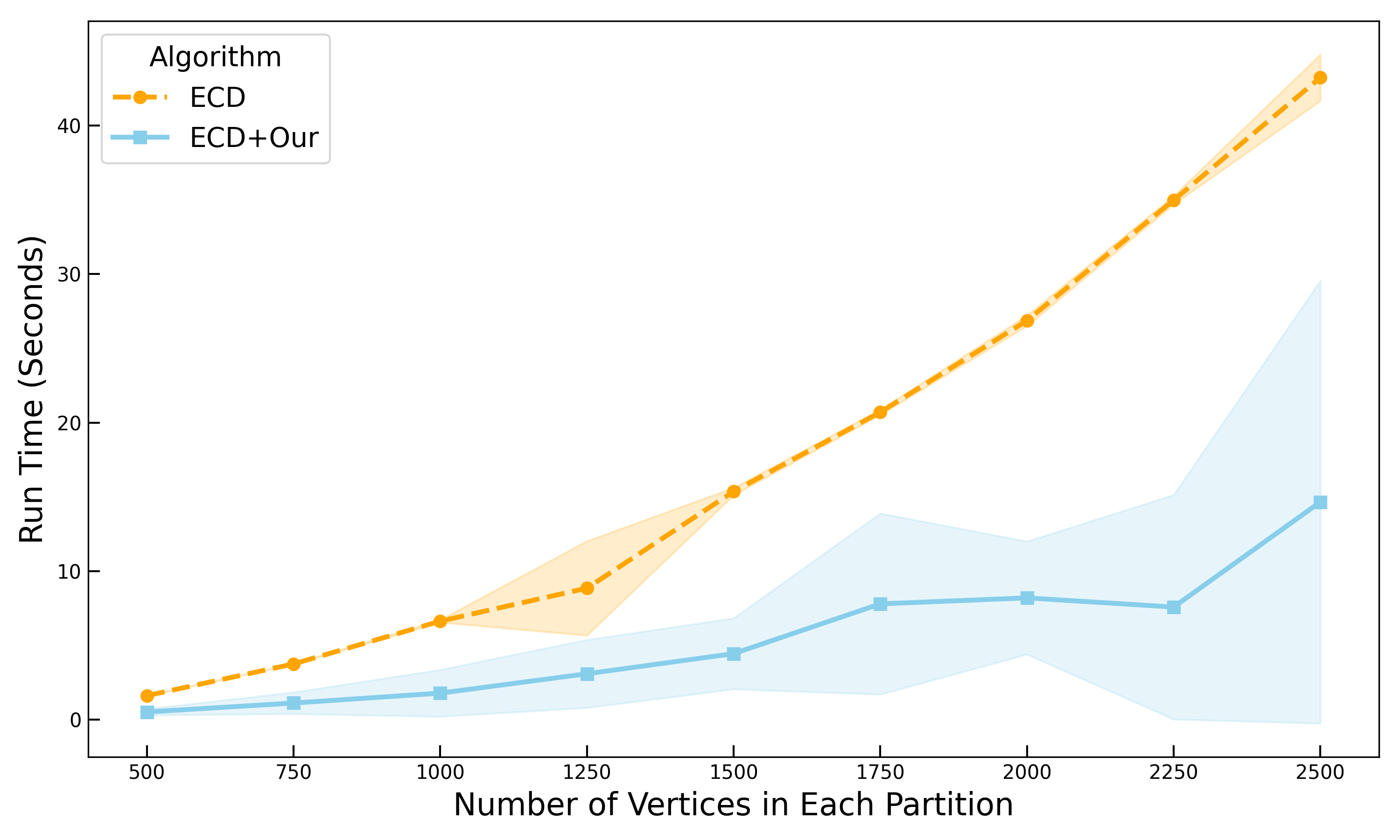}  
    \caption{for $\eta = 0.8$}
    \label{SUBFIGURE LABEL 23}
\end{subfigure}
\begin{subfigure}{.32\textwidth}
    \centering
    \includegraphics[width=.95\linewidth]{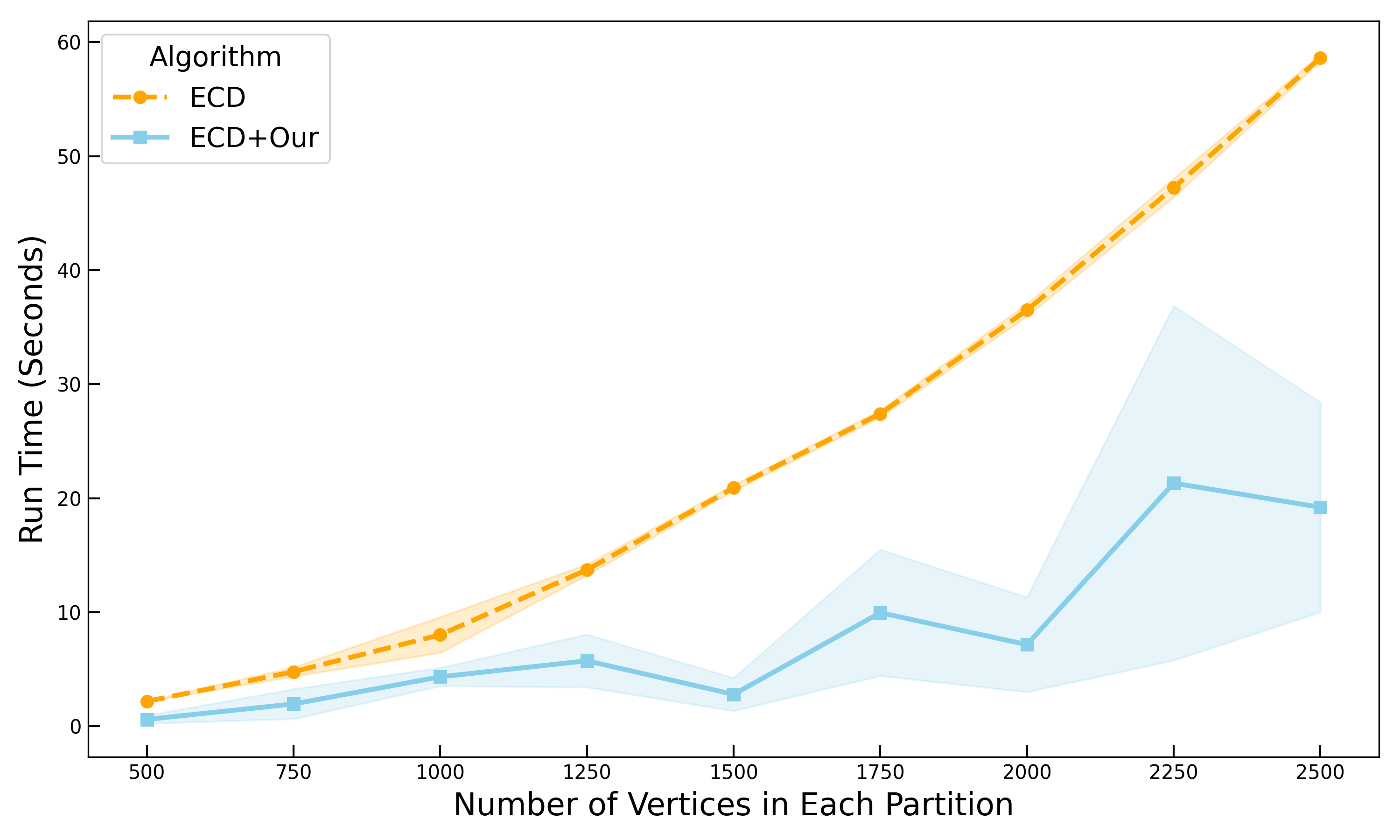}  
    \caption{$\eta = 0.9$}
    \label{SUBFIGURE LABEL 3}
\end{subfigure}
\caption{Runtime comparison between \textsf{ECD} and our algorithm for $\eta = 0.7,0.8$, and $0.9$ respectively.\label{fig:fig-5} }
\end{figure*}

\begin{table*}[t]
\centering
\caption{Comparison of \textsc{ECD} with our algorithm on real-world datasets. We use the vertices corresponding to the listed counties in the first column as the seed vertex of the local algorithm.}
	\begin{tabular}{cccccc}
		\hline
		\textsc{County Name}         & \textsc{Target} $\phi$ & \textsc{ECD}     & \textsc{ECD}        & \textsc{Our Algo.} & \textsc{Our Algo.}    \\  
		&               & \textsc{Runtime} & \textsc{Flow-ratio} & \textsc{Runtime} & \textsc{Flow-ratio} \\ \hline
		Maricopa County      
		& 0.2           & 20.661  & 0.414      & 13.434  & 0.417      \\ \hline
		Virginia Beach City  
		& 0.2           & 15.31   & 0.546      & 12.29   & 0.621      \\ \hline
		Kanawha county       
		& 0.2           & 9.318   & 0.33       & 8.483   & 0.33       \\ \hline
	\end{tabular}
  \label{tab:1}
\end{table*}	

\subsection{Results for Directed Graphs}		
%Next we evaluate the performance of our algorithm for digraphs on synthetic and real-world datasets.

\noindent \textbf{Synthetic Dataset.} We compare the performance of our algorithm with the     \textsc{EvoCutDirected} algorithm presented in \cite{MS21}, which we refer to as \textsc{ECD}, and use the graphs generated from the SBM as the algorithms' input. In our algorithm, given a digraph $G$ as input, we sparsify the graph along with generating the volume-biased ESP on $G’s$ semi-double cover $H$. Since the \textsc{ECD} is a local algorithm, we also test our algorithm locally. In this model, we look into a cluster which is almost bipartite with the bipartition being $L$ and $R$. We set the number of vertices in $L$ and $R$ to be $n_1$ and $n_2$ such that $n_1 = n_2$ and the probability of assigning an edge is defined by  
\[
\begin{blockarray}{ccc}
 & L & R \\
\begin{block}{c(cc)}
L & 9/n_1 & \eta \\
R & 1-\eta & 9/n_2 \\
\end{block}
\end{blockarray},
\]
i.e., the probability that there is an edge within the partition is $9/n_1 = 9/n_2$ and so on. As most of our directed edges are from $L$ to $R$, the value of $\eta$ is high. For our experiments we generate two sets of plots:
\begin{itemize}
\item We   fix the value of $\eta = 0.7$ and increase the number of vertices in each partition from $500$ to $2,500$, and compare the runtime of   \textsc{ECD}   and our algorithm.  As reported in Figure~\ref{fig:fig-4},  our algorithm takes much less time than the \textsc{ECD} algorithm and gives a similar flow-ratio at the same time as we increase the number of vertices.
    
    \item  We increase  the number of vertices in each partition from $500$ to $2,500$, increase  the value of $\eta$ from $0.7$ to $0.9$, and compare the runtime of the \textsc{ECD} algorithm and our algorithm.   As reported in Figure~\ref{fig:fig-5},   our algorithm runs faster than the \textsc{ECD} algorithm as $\eta$ increases.
\end{itemize}

\noindent \textbf{Real-world Dataset.}	Now we evaluate the algorithms’ performance on the US Migration Dataset~\cite{census2000}.   
We construct the digraph from the dataset as  follows: every county in the mainland USA is represented by a vertex; for any vertices $i,j$, the edge weight of is given by $|(M_{i,j}-M_{j,i})/(M_{i,j}+M_{j,i})|$, where $M_{i,j}$ is the number of people who migrated from county to county between 1995 and 2000; in addition, the direction of $(i,j)$ is set to be from $i$ to $j$ if $M_{i,j} > M_{j,i}$, otherwise the direction is set to be the opposite. 
We compare the output of \textsc{ECD} and the output of \textsc{ECD} when applying our sparsification algorithm as subroutine. 
 Furthermore, we use the vertices corresponding to different counties as the input of the local algorithm \textsc{ECD}.
As shown in Table~\ref{tab:1},    with our developed  algorithm the local \textsc{ECD} algorithm achieves roughly the same flow ratio, and our sparsification procedure significantly speeds up the total running time of the algorithm. Moreover, the runtime speedup is more significant when the local algorithm explores more parts of the input graph.	

In conclusion, these experimental studies demonstrate that our developed algorithms can be directly employed to speed up the running time of existing algorithms that find bipartite-like clusters, and can be widely applied when analysing datasets of various domains. We believe that our developed techniques and algorithms will motivate future  and fruitful studies for
analysing complex cluster structures of graphs.

\section*{Impact Statement}

This paper presents work whose goal is to advance the field of Machine Learning. There are many potential societal consequences of our work, none  of which we feel must be specifically highlighted here.

\section*{Acknowledgements}

The first and third authors of the paper are supported by  EPSRC Early Career Fellowship (EP/T00729X/1).

\bibliography{reference}
\bibliographystyle{icml2025}

\newpage
\appendix
\onecolumn

\section{Useful Inequalities}

The following inequalities will be used in our analysis.

	\begin{theorem}[Courant-Fischer Theorem]\label{thm:CF}
		Let $A$ be a $n \times n$ symmetric matrix with eigenvalues $\lambda_1 \le \lambda_2 \le \cdots \le \lambda_n$. Then, it holds for any $1\leq k\leq n$ that 
        \begin{align*}
            \lambda_k & = \min_{\substack{S \\ \dim(S) = k}} \max_{y \in S \setminus \{\mathbf{0}\}} \frac{y^\intercal \cdot A \cdot y}{y^\intercal \cdot y}\\
            &= \max_{\substack{S \\ \dim(S) = n-k+1}} \min_{y \in S\setminus \{\mathbf{0}\}}\frac{y^\intercal \cdot A \cdot y}{y^\intercal \cdot y},
        \end{align*}
		where the maximisation and minimisation are over the  subspaces of $\mathbb{R}^n$.
	\end{theorem}

	\begin{lem}[Bernstein's Inequality]\label{lem:Bern}
		Let $X_1,X_2,\cdots,X_n$ be independent random variables such that $|X_i| \le M$ for any $1 \le i \le n$. Let $X = \sum_{i=1}^n X_i$,  and   $R =  \sum_{i=1}^n \E\left[X_i^2\right]$. Then, it holds that $$\mathbf{P} \left[|X-\E[X]| \ge t\right] \le 2 \exp \left(- \frac{t^2}{2\left(R+\frac{Mt}{3}\right)}\right).$$
	\end{lem}
	
	\begin{lem}[Matrix Chernoff Bound]\label{lem:Chernoff}
		Consider a finite sequence $\{X_i\}$ of independent, random, PSD matrices of dimension $d$ that satisfy $\|X_i\| \le R$. Let $\mu_{\min} = \lambda_{\min} \left(\E[\sum_i X_i]\right)$ 
		and $\mu_{\max} = \lambda_{\max} \left(\E[\sum_i X_i]\right)$. Then, it holds that
		$$
		\mathbf{P} \left[ \lambda_{\min} \left(\sum_i X_i\right) \le (1-\delta) \mu_{\min} \right] \le d \cdot \left(\frac{\mathrm{e}^{-\delta}}{(1-\delta)^{1-\delta}}\right)^{\frac{\mu_{\min}}{R}},
		$$
        for $\delta \in [0,1]$, and
		$$
		\mathbf{P} \left[ \lambda_{\max} \left(\sum_i X_i\right) \ge (1+\delta) \mu_{\max} \right] \le d \cdot \left(\frac{\mathrm{e}^{\delta}}{(1+\delta)^{1+\delta}}\right)^{\frac{\mu_{\max}}{R}} 
		$$
        for $\delta \ge 0$.
	\end{lem}

\section{Omitted Detail from Section~\ref{sec:undirected}\label{sec:proof1}}
This section presents all the omitted detail from Section~\ref{sec:undirected}, and gives a complete proof of Theorem~\ref{thm:undirected}.
 We first recall that, for every vertex $u$ and its adjacent vertex $v$, the algorithm assigns  the edge $e=\{u, v\}$ the probability 
	\begin{equation} 
		p_u(v) \triangleq \min \left\{w_G(u,v) \cdot \frac{C \cdot  \log^3 n}{d_G(u) \cdot (2 - \lambda_{n-k})},1 \right\},
	\end{equation}
	for a large enough constant $C \in \mathbb{R}_{\ge 0}$. The algorithm checks every edge and samples an edge $e = \{u,v\}$ with probability $p_e$, where $$p_e \triangleq p_u(v)+p_v(u)-p_u(v) \cdot p_v(u).$$ Note that, it is easy to check that $p_e$ satisfies the inequality $$\frac{1}{2} (p_u(v)+p_v(u)) \le p_e \le p_u(v)+p_v(u).$$ We start with an empty set $F$ and gradually store all the sampled edges in $F$, which is sampled by the algorithm. Finally, the algorithm returns a weighted graph $H = (V,F,w_H)$, where the weight $w_H(u,v)$ of every sampled edge $e = \{u,v\} \in F$ is defined by $$w_H(u,v) = \frac{w_G(u,v)}{p_e}.$$ 
	
	Next, we analyze the size of $F$. Since $$
	\sum_{u  } \sum_{e = \{u,v\}} w_G(u,v) \cdot \frac{C \cdot  \log^3 n}{d_G(u) \cdot (2 - \lambda_{n-k})} = O \left( \frac{n\cdot   \log^3 n}{2 - \lambda_{n-k}}\right),
	$$  it holds by Markov inequality that with constant probability the number of edges $e = \{u,v\}$ with $p_u(v) \ge 1$ is $O \left( \frac{n \cdot  \log^3 n}{2 - \lambda_{n-k}}\right)$.  Without loss of generality, we assume that these edges are in $F$, and in the remaining part of the proof we assume it holds for any edge $u \sim v$ that $$ w_G(u,v) \cdot \frac{C \cdot \log^3 n}{d_G(u) \cdot (2 - \lambda_{n-k})} < 1.$$ Then, the expected number of edges in $H$ equals 
 \begin{align*}
	\sum_{e = \{u,v\}} p_e & \le \sum_{e = \{u,v\}} p_u(v)+p_v(u)\\
    & = \frac{C \cdot  \log^3 n}{(2 - \lambda_{n-k})}  \sum_{e = \{u,v\}} w(u,v) \cdot \left(\frac{1}{d_G(u)} + \frac{1}{d_G(v)}\right) \\
    & = O \left( \frac{n \cdot  \log^3 n}{2 - \lambda_{n-k}}\right),
	\end{align*} and by Markov inequality it holds with constant probability that $$|F| = O \left( \frac{n\cdot  \log^3 n}{2 - \lambda_{n-k}}\right).$$
	
	Now we   show that the cut value between $A_i$ and $B_i$ is preserved in $H$ for all $1 \le i \le k$.  For any edge $e = \{u,v\}$, we define the  random variable $Y_e$ by 
	\begin{equation}\label{eqn:Ye}
		Y_e = \begin{cases}
			\dfrac{w_G(u,v)}{p_e} & \text{with probability } p_e, \\
			0 & \text{otherwise.}
		\end{cases}
	\end{equation}
	Also, we define $X = w_H(A_i,B_i)$, and have that  
	\begin{equation}
    \begin{split}
        \E[X] & = \sum_{\substack{e = \{u,v\} \\ u \in A_i, v \in B_i}} \E\left[Y_e\right] = \sum_{\substack{e = \{u,v\} \\ u \in A_i, v \in B_i}} p_e \cdot \dfrac{w_G(u,v)}{p_e} \\
        & = \sum_{\substack{e = \{u,v\} \\ u \in A_i, v \in B_i}} w_G(u,v) = w_G(A_i,B_i).
    \end{split}
	\end{equation}
	Next, we analyse the second moment of the random variable $X$ and have that 
	\begin{align}
		\begin{split}
			\E\left[X^2\right] &= \sum_{\substack{e = \{u,v\} \\ u \in A_i, v \in B_i}} p_e \cdot \left(\dfrac{w_G(u,v)}{p_e}\right)^2\\
            &= \sum_{\substack{e = \{u,v\} \\ u \in A_i, v \in B_i}} \dfrac{w_G(u,v)^2}{p_e}\\
			&\le \sum_{\substack{e = \{u,v\} \\ u \in A_i, v \in B_i}} \dfrac{2w_G(u,v)^2}{p_u(v)+p_v(u)}\\
			&= \sum_{\substack{e = \{u,v\} \\ u \in A_i, v \in B_i}} \dfrac{2w_G(u,v)^2}{ \frac{w_G(u,v) \cdot C \cdot  \log^3 n}{(2 - \lambda_{n-k})} \cdot \left(\frac{1}{d_G(u)} + \frac{1}{d_G(v)}\right)}\\
			&\le \frac{2- \lambda_{n-k}}{C \cdot  \log^3 n} \ds \sum_{\substack{e = \{u,v\} \\ u \in A_i, v \in B_i}} w(u,v) \cdot \left(\frac{d_G(u)+d_G(v)}{2}\right),
		\end{split}
	\end{align}
	where the last step follows by the means inequality.
	Let $\{(A_i,B_i)\}_{i=1}^k$ be the optimal cluster where $\bar{\rho}(k)$ is attained for graph $G$. Recall that for every $k \in \mathbb{N}$, the $k$-way dual Cheeger constant is defined by 
	\begin{equation*}
		\bar{\rho}_G(k) = \max_{(A_1,B_1),\cdots,(A_k,B_k)} \min_{1 \le i \le k} \overline{\phi}_G (A_i,B_i).
	\end{equation*}
	Then, we have  for every $1 \le i \le k$ that  $$\bar{\rho}_G(k) \le \overline{\phi}_G (A_i,B_i) = \frac{2 w_G(A_i,B_i)}{\vol_G(A_i \cup B_i)},$$ which implies
	\begin{equation}
		\frac{\bar{\rho}_G(k)}{2} \cdot \vol_G(A_i \cup B_i) \le \sum_{\substack{e = \{u,v\} \\ u \in A_i, v \in B_i}} w_G(u,v).
	\end{equation}
	Next, by the  Chebyshev's inequality we have for any constant $c\in\mathbb{R}^+$ that 
	\begin{align}\label{eqn:Cheb}
		\begin{split}
		\lefteqn{\textbf{P}\left[|X-\E[X]| \ge c \cdot \E[X]\right]}\\
   &\le \frac{\E[X^2]}{c^2 \cdot \E[X]^2}\\
			&\le \frac{\frac{2- \lambda_{n-k}}{C \cdot  \log^3 n} \left( \sum_{\substack{e = \{u,v\} \\ u \in A_i, v \in B_i}} w_G(u,v) \cdot \left(\frac{d_G(u)+d_G(v)}{2}\right) \right)}{c^2 \cdot \left( \sum_{\substack{e = \{u,v\} \\ u \in A_i, v \in B_i}} w_G(u,v)\right)^2}   \\
			&\le \frac{\frac{2- \lambda_{n-k}}{C \cdot  \log^3 n} \left( \sum_{\substack{e = \{u,v\} \\ u \in A_i, v \in B_i}} w_G(u,v) \cdot \left(\frac{d_G(u)+d_G(v)}{2}\right) \right)}{c^2 \cdot \left( \frac{\bar{\rho}_G(k)}{2} \cdot \vol_G(A_i \cup B_i)\right)^2}\\
			&= \frac{ 2\cdot (2- \lambda_{n-k})}{c^2 \cdot C\cdot   \log^3 n \cdot \bar{\rho}_G(k)^2} \cdot \frac{\sum_{\substack{e = \{u,v\} \\ u \in A_i, v \in B_i}} w_G(u,v) \cdot \left(d_G(u)+d_G(v)\right)}{\vol_G(A_i \cup B_i)^2}\\
			&\le \frac{2\cdot (2- \lambda_{n-k})}{c^2 \cdot C\cdot  \log^3 n \cdot \bar{\rho}_G(k)^2} \cdot \left(\max_{\substack{e = \{u,v\} \\ u \in A_i, v \in B_i}} \{d_G(u)+d_G(v)\}\right) \cdot \frac{ \sum_{\substack{e = \{u,v\} \\ u \in A_i, v \in B_i}} w_G(u,v)}{\vol_G(A_i \cup B_i)^2}.
		\end{split}
	\end{align}
	Since $\vol_G(A_i \cup B_i) = \sum_{u \in A_i} d_G(u) + \sum_{v \in B_i} d_G(v)$ and $d_G(u) = \sum_{u \sim v} w_G(u,v)$, we have
    \begin{align*}
        \max_{\substack{e = \{u,v\} \\ u \in A_i, v \in B_i}} \left\{d_G(u)+d_G(v)\right\} & \le \sum_{u \in A_i} d_G(u) + \sum_{v \in B_i} d_G(v)\\
        &= \vol_G(A_i \cup B_i)
    \end{align*}
    and $$
	\sum_{\substack{e = \{u,v\} \\ u \in A_i, v \in B_i}} w_G(u,v) \le \vol_G(A_i \cup B_i).
	$$
	Thus, we have  by  \eqref{eqn:Cheb} and the assumption of $\bar{\rho}(k)\geq \frac{1}{\log(n)}$  that
    \begin{align*}
        \textbf{P}\left[|X-\E[X]| \ge c \cdot \E[X]\right] &\le \frac{2(2- \lambda_{n-k})}{c^2 \cdot C \cdot  \log^3 n \cdot \bar{\rho}(k)^2}\\ &= O\left(\frac{1}{ \log n}\right).
    \end{align*}
	Hence, by choosing a sufficient large constant $c$ and the union bound, we have that 	\begin{equation}\label{eqn:cutH}
		w_H(A_i,B_i) = \Omega\left(w_G(A_i,B_i)\right) \text{ for all } 1\le i \le k.
	\end{equation}

	Next, we show that the degree of every vertex  in $H$ is approximately preserved with high probability. Based on the random variable $Y_e$ defined in \eqref{eqn:Ye}, we define the random variable $Z_u$ by $$Z_u = \sum_{e:v\sim u} Y_e.$$ Then, the expected value of $Z_u$ is given by 
    \begin{align*}
        \E[Z_u] &= \sum_{e:v\sim u} \E[Y_e] = \sum_{e:v\sim u} p_e \cdot \frac{w_G(u,v)}{p_e}\\
        &= \sum_{v :v\sim u} w_G(u,v) = d_G(u),
    \end{align*}
	and the second moment can be upper bounded by
    \begin{align*}
        \sum_{e:v\sim u} \E\left[Y_e^2\right] &=  \sum_{e:v\sim u} p_e \cdot \left(\frac{w_G(u,v)}{p_e}\right)^2\\
        & = \sum_{e:v\sim u} \frac{w_G(u,v)^2}{p_e} \le \sum_{v:v\sim u} \frac{w_G(u,v)^2}{p_u(v)},
    \end{align*}
	since $p_e \ge p_u(v).$ Now using the value of $p_u(v)$ from \eqref{eqn:pu}, we have
	\begin{align*}
	\sum_{e:v\sim u} \E\left[Y_e^2\right] & \le  \sum_{v:v\sim u} w(u,v)^2 \cdot \frac{d_G(u) \cdot (2 - \lambda_{n-k})}{w(u,v) \cdot C \cdot  \log^3 n}\\
    & = \frac{d_G(u) \cdot (2 - \lambda_{n-k})}{C \cdot \log^3 n} \sum_{ v:v\sim u} w_G(u,v)\\
    & = \frac{d_G^2(u) \cdot (2 - \lambda_{n-k})}{C \cdot  \log^3 n}
	\end{align*}
	and for any edge $e = \{u,v\}$ we have that
	$$
	0 \le \frac{w(u,v)}{p_e} \le \frac{w(u,v)}{p_u(v)} \le \frac{d_G(u) \cdot (2 - \lambda_{n-k})}{C \cdot  \log^3 n}.
	$$
	Now, applying Bernstein's inequality (Lemma~\ref{lem:Bern}), we have  
	\begin{align*}
		& \textbf{P} \left[|d_H(u)-d_G(u)| \ge \frac{d_u}{2}\right] \\
        &= \textbf{P} \left[|Z_u-E[Z_u]| \ge \frac{\E[Z_u]}{2}\right]\\
		&\le 2 \cdot \exp \left(\frac{-\frac{1}{8} \cdot d_G^2(u)}{\frac{d_G^2(u) \cdot (2 - \lambda_{n-k})}{C \cdot  \log^3 n} + \frac{1}{6} \cdot \frac{d_G^2(u) \cdot (2 - \lambda_{n-k})}{C \cdot  \log^3 n}} \right)\\
		&= 2 \cdot \exp \left(- \frac{\frac{1}{8} \cdot C \cdot  \log^3 n}{\frac{7}{6} \cdot (2 - \lambda_{n-k})}\right)\\
		&=   o \left( \frac{1}{n^2}\right).
	\end{align*} 
	Hence, it holds by the union bound that, with high probability, the degree of all the vertices in $H$ are approximately preserved up to a constant factor. This implies that for any subset $S \subseteq V$, we have
	\begin{equation*}
		\vol_H(S) = \Theta \left( \vol_G (S)\right),
	\end{equation*}
	more specifically,
	\begin{equation}\label{eqn:volH}
		\vol_H(A_i \cup B_i) = \Theta \left( \vol_G (A_i \cup B_i)\right),
	\end{equation}
	for all $1\le i \le k.$	Thus, combining  \eqref{eqn:cutH} and \eqref{eqn:volH} gives us that 
	\begin{equation}
		\overline{\phi}_H (A_i,B_i) = \Omega \left( \overline{\phi}_G (A_i,B_i)\right)
	\end{equation}
	for all $1\le i \le k$, which implies that
    \begin{align*}
        \bar{\rho}_H(k) \ge \min_{1 \le i \le k} \overline{\phi}_H (A_i,B_i) & = \min_{1 \le i \le k} \Omega \left(  \overline{\phi}_G (A_i,B_i)\right)\\
        &= \Omega \left(  \bar{\rho}_G(k)\right),
    \end{align*}
	where the last equality follows from the fact that $\{(A_i,B_i)\}_{i=1}^k$ is the optimal cluster where $\bar{\rho}(k)$ is attained for graph $G$.
	
	Next, we show that the top $(n-k)$-eigenspaces of  $\mathcal{J}_G$ are preserved in $H$. Without loss of generality we assume the graph is connected.  Since $\mathcal{J}_G = 2I -\mathcal{L}_G$ by definition, it holds that
	\begin{equation}\label{eq:eigen_relation}
		\lambda_{i}(\mathcal{J}_G) = 2 - \lambda_{n+1-i}(\mathcal{L}_G).
	\end{equation}
	%	For simplicity, let the eigenvalues of $2I-\mathcal{L}_G$ be \he{to change the notation} 	$$ 	\gamma_1 \le \gamma_2 \le \cdots \le \gamma_n, 	$$ 	where $\gamma_i = 2 - \lambda_{n+1-i}$.
	Let
	$$
	\mathcal{P} \triangleq \sum_{i=1}^{n-k} (2-\lambda_i(\mathcal{L}_G)) f_i f_i^\intercal,
	$$
	and with slight abuse of notation we call $\mathcal{P}^{-1/2}$ as the square root of the pseudo-inverse of $\mathcal{P}$, i.e., 
	$$
	\mathcal{P}^{-1/2} = \sum_{i=1}^{n-k} (2-\lambda_i(\mathcal{L}_G) )^{-1/2} f_i f_i^\intercal.
	$$
	Let $\overline{\mathcal{P}}$ be the projection on the span of $\left\{f_1,f_2,\cdots,f_{n-k}\right\}$, then
	$$
	\overline{\mathcal{P}} =  \sum_{i=1}^{n-k} f_i f_i^\intercal.
	$$
	Recall that, for each vertex $v$, the indicator vector $\chi_v \in \mathbb{R}^n$ is defined by $\chi_v(u) = \frac{1}{\sqrt{d_G(v)}}$ if $u = v$ and $\chi_v(u) = 0$ otherwise. For each edge $e=\{u,v\}$ of $G$ we define a vector $g_e = \chi_u+\chi_v \in \mathbb{R}^n$ and a random matrix $X_e \in \mathbb{R}^{n \times n}$ by
	\begin{equation}
		X_e = \begin{cases}
			w_H(u,v) \cdot \mathcal{P}^{-1/2} g_e g_e^\intercal \mathcal{P}^{-1/2} & \text{if } e=\{u,v\} \text{ is sampled}\\
            & \text{by the algorithm,}\\
			\textbf{0} & \text{otherwise.}
		\end{cases}
	\end{equation}
	Then, it holds that 
	\begin{align*}
		\sum_{e \in E} X_e &= \sum_{e=\{u,v\} \in F} w_H(u,v) \cdot \mathcal{P}^{-1/2} g_e g_e^\intercal \mathcal{P}^{-1/2}\\
		&=  \mathcal{P}^{-1/2} \left(\sum_{e=\{u,v\} \in F}  w_H(u,v) \cdot g_e g_e^\intercal \right)  \mathcal{P}^{-1/2} \\
		&= \mathcal{P}^{-1/2} \mathcal{J}_H' \mathcal{P}^{-1/2},
	\end{align*}
	where $$ \mathcal{J}_H' \triangleq \sum_{e=\{u,v\} \in F}  w_H(u,v) \cdot g_e g_e^\intercal$$ is   the signless Laplacian matrix of $H$ normalised with respect to the degree of the vertices in the original graph $G$. We will now prove that, with high probability the top   $n-k$ eigenspaces of $\mathcal{J}_H'$ and $\mathcal{J}_G$ are approximately the same. We first analyse the   expectation of $\sum_{e \in E} X_e$, and have that 
	\begin{align*}
		\E\left[\sum_{e \in E} X_e\right] &= \sum_{e=\{u,v\} \in E} p_e \cdot w_H(u,v) \cdot \mathcal{P}^{-1/2} g_e g_e^\intercal \mathcal{P}^{-1/2}\\
		&=  \sum_{e=\{u,v\} \in E} p_e \cdot \frac{w_G(u,v)}{p_e} \cdot \mathcal{P}^{-1/2} g_e g_e^\intercal \mathcal{P}^{-1/2}\\
		&= \mathcal{P}^{-1/2} \left(\sum_{e=\{u,v\} \in F}  w_G(u,v) \cdot g_e g_e^\intercal\right) \mathcal{P}^{-1/2}\\
		&= \mathcal{P}^{-1/2}  \mathcal{J}_G \mathcal{P}^{-1/2} = \sum_{i=1}^{n-k} f_i f_i^\intercal = \overline{\mathcal{P}}.
	\end{align*}
	Moreover, for any edge $e=\{u,v\} \in E$   sampled by the algorithm, we have
	\begin{align*}
		\|X_e\| &\le w_H(u,v) \cdot g_e^\intercal \mathcal{P}^{-1/2 }   \mathcal{P}^{-1/2} g_e =  \frac{w_G(u,v)}{p_e} \cdot  g_e^\intercal \mathcal{P}^{-1} g_e\\
		&\le \frac{w_G(u,v)}{p_e} \cdot \frac{1}{2-\lambda_{n-k}}\cdot  \|g_e\|^2\\
		&\le \frac{2w_G(u,v)}{p_u(v)+p_v(u)} \cdot \frac{1}{2-\lambda_{n-k}} \cdot \left(\frac{1}{d_G(u)} + \frac{1}{d_G(v)}\right) \\
		&\le \frac{2}{C \cdot  \log^3 n},
	\end{align*}
	where the second inequality follows by the min-max theorem of eigenvalues. Now we apply the matrix Chernoff bound (Lemma~\ref{lem:Chernoff}) to analyze the eigenvalues of $\sum_{e \in E} X_e$. Following Lemma~\ref{lem:Chernoff} we set the parameters as follows:
	\begin{align}
		\begin{split}
			\mu_{\max} &= \lambda_{\max} \left(\E\left[\sum_{e \in E} X_e\right]\right) = \lambda_{\max} \left(\overline{\mathcal{P}}\right) = 1,\\
			R &= \frac{2}{C \cdot  \log^3 n}, \text{ and }\\
			\delta &= \frac{1}{2}.
		\end{split}
	\end{align}
	Then using the Matrix Chernoff bound (Lemma~\ref{lem:Chernoff}), we have
	$$
	\textbf{P} \left[ \lambda_{\max} \left(\sum_{e \in E} X_e\right) \ge \frac{3}{2} \right] \le n \cdot \left(\frac{\mathrm{e}^{\frac{1}{2}}}{1.5^{\frac{3}{2}}}\right)^{\frac{C \cdot  \log^3 n}{2}} = O\left(\frac{1}{ n^3}\right),
	$$
	for some constant $C$; this implies that \begin{equation}\label{eqn:lambdamax}
		\textbf{P} \left[ \lambda_{\max} \left(\sum_{e \in E} X_e\right) \le \frac{3}{2} \right] = 1 - O\left(\frac{1}{n^3}\right).
	\end{equation}
	On the other hand, since $\E\left[\sum_{e \in E} X_e\right] = \overline{\mathcal{P}}$, we have $\mu_{\min} = 1$ and hence keeping $R$ and $\delta$ the same as above, using the Matrix Chernoff bound (Lemma~\ref{lem:Chernoff}), we get
	$$
	\textbf{P} \left[ \lambda_{\min} \left(\sum_{e \in E} X_e\right) \le \frac{1}{2} \right] \le n \cdot \left(\frac{\mathrm{e}^{-\frac{1}{2}}}{0.5^{\frac{1}{2}}}\right)^{\frac{C \cdot  \log^3 n}{2}} = O\left( \frac{1}{n^3}\right);
	$$
	this implies that 
	\begin{equation}\label{eqn:lambdamin}
		\textbf{P} \left[ \lambda_{\min} \left(\sum_{e \in E} X_e\right) \geq \frac{1}{2} \right] = 1 - O\left( \frac{1}{n^3}\right).
	\end{equation}
	Combining \eqref{eqn:lambdamax}, \eqref{eqn:lambdamin}   and the fact that $\sum_{e \in E} X_e = \mathcal{P}^{-1/2} \mathcal{J}_H' \mathcal{P}^{-1/2}$, with probability $1 -  O\left(\frac{1}{n^3}\right)$ it holds for any non-zero $x \in \mathbb{R}^n$ in span$\{f_1,f_2,\cdots,f_{n-k}\}$ that
	\begin{equation}\label{eqn:1/2,3/2}
		\frac{x^\intercal \mathcal{P}^{-1/2}  \mathcal{J}_H' \mathcal{P}^{-1/2}   x}{x^\intercal  x} \in \left[\frac{1}{2},\frac{3}{2}\right].
	\end{equation}
	Let $y = \mathcal{P}^{-1/2}   x$, and  we   rewrite  \eqref{eqn:1/2,3/2} as
	$$
	\frac{y^\intercal  \mathcal{J}_H' y}{y^\intercal \mathcal{P} y} = \frac{y^\intercal  \mathcal{J}_H' y}{y^\intercal y} \cdot \frac{y^\intercal y}{y^\intercal \mathcal{P} y} \in \left[\frac{1}{2},\frac{3}{2}\right].
	$$
	Since $\dim(\textup{span} \{f_{1},f_2,\cdots,f_{n-k}\}) = n-k$, there exist $n-k$ orthogonal vectors   whose Rayleigh quotient with respect to $\mathcal{J}_H'$ is $\Theta(\lambda_{n-k}(2I-\mathcal{L}_G))$. Hence, by the Courant-Fischer Theorem (Theorem~\ref{thm:CF})  we have	  \begin{equation}\label{eqn:LH'-LG}
		\frac{1}{2} \cdot \lambda_{n-k}(2I-\mathcal{L}_G) \le   \lambda_{k+1}(\mathcal{J}_H') \le \frac{3}{2} \cdot \lambda_{n-k}(2I-\mathcal{L}_G)
	\end{equation}
	By the definition of  $\mathcal{J}_H' = D_G^{-1/2} \left(D_H + A_H\right)D_G^{-1/2}$, we have  
    \begin{align*}
    \mathcal{J}_H &= D_H^{-1/2} \left(D_H + A_H\right)D_H^{-1/2}\\
    &= D_H^{-1/2} \left(D_G^{1/2} \cdot  \mathcal{J}_H' \cdot D_G^{1/2}\right)D_H^{-1/2}.
    \end{align*}
	Hence, we set $y = D_G^{1/2}   D_H^{-1/2}  x$ for any $x \in \mathbb{R}^n$ and have that 
	\begin{equation}\label{eqn:xy}
    \begin{split}
        \frac{x^\intercal     \mathcal{J}_H   x}{x^\intercal \cdot  x} & = \frac{x^\intercal \cdot D_H^{-1/2} \left(D_G^{1/2} \cdot  \mathcal{J}_H'  \cdot D_G^{1/2}\right)D_H^{-1/2} \cdot x}{x^\intercal \cdot x}\\
        & = \frac{y^\intercal \cdot  \mathcal{J}_H' \cdot y}{x^\intercal \cdot x} \ge \frac{1}{2}\cdot \frac{y^\intercal \cdot  {\mathcal{J}_H'} \cdot y}{y^\intercal \cdot y},
    \end{split}
	\end{equation}
	where we use the fact that the  degree of a vertex differs by a constant factor between $H$ and $G$. Similarly, we  also have
	\begin{equation}\label{eqn:xy-1}
		\frac{x^\intercal \cdot \mathcal{J}_H \cdot x}{x^\intercal \cdot x} \le \frac{3}{2} \cdot \frac{y^\intercal \cdot  \mathcal{J}_H' \cdot y}{y^\intercal \cdot y},
	\end{equation}
	
	Let $T \subset \mathbb{R}^n$ be a $(k+1)$-dimensional subspace of $\mathbb{R}^n$ satisfying $$
	\lambda_{k+1}( \mathcal{J}_H) = \max_{x \ne 0, x \in T} \frac{x^\intercal \cdot  \mathcal{J}_H\cdot x}{x^\intercal \cdot x},
	$$ and   $\widetilde{T} = \left\{D_G^{1/2}   D_H^{-1/2}   x : x \in T\right\}$. Since   $D_G^{1/2}   D_H^{-1/2}$ has full rank, $\widetilde{T}$ is also a $(k+1)$-dimensional subspace of $\mathbb{R}^n$. Hence, by  the Courant-Fischer Theorem~(Theorem~\ref{thm:CF}) and  \eqref{eqn:xy}, we have that   
	\begin{align}\label{eqn:LH'-LH}
		\begin{split}
			\lambda_{k+1}(\mathcal{J}_H') &= \min_{\substack{S \\ \dim(S) = k+1}} \max_{y \in S \setminus \{\mathbf{0}\}} \frac{y^\intercal \cdot  \mathcal{J}_H' \cdot y}{y^\intercal \cdot y}\\
			&\le \max_{y \in \widetilde{T} \setminus \{\mathbf{0}\}} \frac{y^\intercal \cdot  \mathcal{J}_H' \cdot y}{y^\intercal \cdot y} \\
			&\le 2 \cdot \max_{x \in T \setminus \{\mathbf{0}\}} \frac{x^\intercal \cdot \mathcal{J}_H \cdot x}{x^\intercal \cdot x}  = 2 \cdot \lambda_{k+1}(\mathcal{J}_H).
		\end{split}
	\end{align}
	Next, using  \eqref{eqn:LH'-LG} and \eqref{eqn:LH'-LH}, we have
	$$
	\frac{1}{2}\cdot \lambda_{k+1}(\mathcal{J}_G) \le \lambda_{k+1}(\mathcal{J}_H') \le 2 \cdot \lambda_{k+1}(\mathcal{J}_H),
	$$ which implies that 
	\begin{equation}\label{eqn:LG-LH-1}
		\frac{1}{4}\cdot  \lambda_{k+1}(\mathcal{J}_G) \le \lambda_{k+1}( \mathcal{J}_H).
	\end{equation}
	
	Similarly, let $U \subset \mathbb{R}^n$ be an $(n-k)$-dimensional subspace of $\mathbb{R}^n$ satisfying $$
	\lambda_{k+1}( \mathcal{J}_H)  = \min_{x \ne 0, x \in U} \frac{x^\intercal \cdot \mathcal{J}_H \cdot x}{x^\intercal \cdot x},
	$$   and  $\widetilde{U} = \left\{D_G^{1/2}  D_H^{-1/2}   x : x \in U\right\}$. Since   $D_G^{1/2} \cdot D_H^{-1/2}$ has full rank, $\widetilde{U}$ is also an $(n-k)$-dimensional subspace of $\mathbb{R}^n$. Thus, using the Courant-Fischer Theorem (Theorem~\ref{thm:CF}) and  \eqref{eqn:xy-1}, we have  
	\begin{align}\label{eqn:LH'-LH-1}
		\begin{split}
			\lambda_{k+1}( \mathcal{J}_H') &= \max_{\substack{S \\ \dim(S) = n-k}} \min_{y \in S \setminus \{\mathbf{0}\}}\frac{y^\intercal \cdot \mathcal{J}_H'  \cdot y}{y^\intercal \cdot y}\\
			&\ge \min_{y \in \widetilde{U} \setminus \{\mathbf{0}\}} \frac{y^\intercal \cdot  \mathcal{J}_H' \cdot y}{y^\intercal \cdot y} \\
			&\ge \frac{2}{3} \cdot \min_{x \in U \setminus \{\mathbf{0}\}} \frac{x^\intercal \cdot \left(2I - \mathcal{L}_H\right) \cdot x}{x^\intercal \cdot x}\\
			& = \frac{2}{3} \cdot \lambda_{k+1} \left( \mathcal{J}_H\right).
		\end{split}
	\end{align} 
	Next, by \eqref{eqn:LH'-LG} and \eqref{eqn:LH'-LH-1} we have 
	$$
	\frac{2}{3} \cdot \lambda_{k+1}( \mathcal{J}_H) \le \gamma_{k+1}(\mathcal{L}_H') \le \frac{3}{2} \cdot \lambda_{k+1}(\mathcal{J}_G),
	$$ which implies that 
	\begin{equation}\label{eqn:LG-LH-2}
		\lambda_{k+1}( \mathcal{J}_H) \le \frac{9}{4} \cdot \lambda_{k+1}(\mathcal{J}_G).
	\end{equation}
	Thus, combining \eqref{eqn:LG-LH-1} and \eqref{eqn:LG-LH-2} we have
	$$
	\frac{1}{4} \cdot \lambda_{k+1}(\mathcal{J}_G) \le \lambda_{k+1}( \mathcal{J}_H) \le \frac{9}{4} \cdot \lambda_{k+1}(\mathcal{J}_G),
	$$
	Hence, the  the top $n-k$ eigenspaces of $ \mathcal{J}_G$ are preserved in $\mathcal{J}_H$. This proves the second statement of the theorem.

\section{Omitted Detail from Section~\ref{sec:directed}}

In this section we list all the proofs omitted from Section~\ref{sec:directed}.

 \begin{proof}[Proof of Lemma~\ref{lem:Fphi}]
  The proof follows from \cite{MS21}, which proves the result for undirected graphs. We include the proof here for completeness. 
  Let $S = A_1 \cup B_2$ in $H$, then 
  \begin{align}
      \begin{split}
          \phi_H(A_1 \cup B_2) = \phi_H(S) &= \frac{w_H(S,V\setminus S)}{\vol_H(S)}\\
          &= \frac{\vol_H(S) - 2w_H(S,S)}{\vol_H(S)}\\
          &= 1 - \frac{2w_H(S,S)}{\vol_H(S)}\\
          &= 1 - \frac{2{w}_{\overrightarrow{G}}(A,B)}{\vol_{\out}(A) + \vol_{\ino}(B)}\\
          &= f_{\overrightarrow{G}}(A,B).
      \end{split}
  \end{align}
  This proves the first statement of the lemma. The second statement of the lemma follows by the similar argument.
 \end{proof}

\begin{proof}[Proof of Lemma~\ref{lem:directed_reduction}]
By definition, we   have that 
  	\begin{equation}\label{eqn:Fphi}
		f_{\overrightarrow{G}}(A,B) = 1 - \overline{\phi}_{\overrightarrow{G}}(A,B),
	\end{equation}
and this implies that 
	\begin{align}
	\bar{\rho}_{\overrightarrow{G}}(k) &= \max_{(A_1,B_1),\ldots,(A_k,B_k)} \min_{1 \le i \le k} \overline{\phi}_{\overrightarrow{G}} (A_i,B_i)   \nonumber \\
		&= \max_{(A_1,B_1),\ldots,(A_k,B_k)} \min_{1 \le i \le k} \left(1 -  f_{\overrightarrow{G}}(A_i,B_i)\right)  \nonumber \\
		&= 1 -  \min_{(A_1,B_1),\ldots,(A_k,B_k)} \max_{1 \le i \le k}   f_{\overrightarrow{G}}(A_i,B_i)  \nonumber  \\
		&= 1 -  \min_{C_1,\ldots,C_k} \max_{1 \le i \le k} \phi_H(C_i),   \nonumber 
	\end{align}
where the second line follow by \eqref{eqn:Fphi}, and the last one follows by Lemma~\ref{lem:Fphi} and $C_i = A_{i_1} \cup B_{i_2}$. 
\end{proof}

\end{document}